\documentclass[a4paper,10pt]{article}
\usepackage[utf8]{inputenc}
\usepackage{amsfonts,amsmath, amssymb, bbm,amsthm}
\usepackage{dsfont}
\usepackage{mathtools}
\usepackage[hidelinks]{hyperref}
\usepackage{tikz}
\usetikzlibrary{arrows.meta}
\usepackage{comment}
\usepackage{mathtools}
\newtheorem{prop}{Proposition}
\newtheorem{thm}{Theorem}

\usepackage{soul}
\usetikzlibrary{arrows.meta}
\usepackage{bm}
\usepackage{cancel}

\usepackage{mathrsfs}

\usepackage{accents}
\newcommand\munderbar[1]{%
  \underaccent{\bar}{#1}}

\usepackage{geometry}
\def\hmath$#1${\texorpdfstring{{\rmfamily\textit{#1}}}{#1}}
\def\al{\alpha}
\def\be{\beta}
\newcommand\cS{{\mathcal S}}
\newcommand\cI{{\mathcal I}}
\newcommand\hh{{f}}
\def\acS{{\, \xrightarrow{\cS}\, }}
\theoremstyle{definition}

\newtheorem{remark}{Remark}

\title{Darboux transformations and related non-Abelian\\ integrable differential-difference systems\\
of the derivative nonlinear Schr\"odinger type }
\author{Edoardo Peroni$^{a}$, Jing Ping Wang$^{b, a}$\\
{\small$^{a}$School of Engineering, Mathematics and Physics, University of Kent, Canterbury, UK,}\\
{\small$^b$ School of Mathematics and Statistics, Ningbo University, Ningbo 315211, People’s Republic of China.}\\ \\
}
\date{}

\begin{document}

\maketitle

\begin{abstract}
We construct linear and quadratic Darboux matrices compatible with the reduction group of the Lax operator for each of the seven known non-Abelian derivative nonlinear Schrödinger equations that admit Lax representations. The differential-difference systems derived from these Darboux transformations generalise established non-Abelian integrable models by incorporating non-commutative constants. Specifically, we demonstrate that linear Darboux transformations generate non-Abelian Volterra-type equations, while quadratic transformations yield two-component systems, including non-Abelian versions of the Ablowitz-Ladik, Merola-Ragnisco-Tu, and relativistic Toda equations. Using quasideterminants, we establish necessary conditions for factorising a higher-degree polynomial Darboux matrix with a specific linear Darboux matrix as a factor. This result enables the factorisation of quadratic Darboux matrices into pairs of linear Darboux matrices.
\end{abstract}

\section{Introduction}
Darboux transformations play a significant role in the study of integrable partial differential equations (PDEs). A Darboux transformation,
acting as an automorphism of the Lax representation of such equation, leads to a B\"acklund transformation, which is widely used to construct exact solutions. Furthermore, Darboux transformations can themselves be interpreted as integrable differential-difference equations (D$\Delta$Es). \par

Originally introduced by Darboux in the context of linear differential equations \cite{9b}, the framework was subsequently extended to nonlinear systems. Over time, Darboux transformations have become a powerful tool in the study of integrable PDEs, with a substantial body of literature dedicated to their applications.  While a full survey of references is beyond our scope, notable contributions relevant to this work include \cite{1, 80, 135, 154, 132, 131}.

Much of the research on Darboux transformations focuses on the dressing method, a procedure that generates new soliton solutions from known ones. In this paper, we analyse the Darboux transformations as integrable systems themselves, represented by D$\Delta$Es.
These equations effectively encode the relationships between solutions produced by the dressing method and can be viewed as discrete analogues of the original PDEs. By further discretising the system and ensuring the compatibility of two Darboux transformations, it is possible to obtain fully discrete equations, known as partial-difference equations (P$\Delta$Es).
This connection—linking PDEs, D$\Delta$Es, and P$\Delta$Es—is often referred to as the Lax-Darboux scheme.
For a detailed introduction to this framework, see \cite{133, 133b}.\par

In this paper, we explore the Lax-Darboux scheme for non-Abelian derivative nonlinear Schr\"odinger (DNLS) equations, constructing their Darboux transformations using a Darboux matrix that is polynomial in the spectral parameter. While aspects of this problem for classical DNLS equations have been addressed in prior works, notably in \cite{80, 179, 132, 103}, few studies have approached it from a non-commutative perspective. Furthermore, to our knowledge, none have treated the involved constants as non-commutative. Our objective is to provide a unified and rigorous treatment of this topic.

The classical DNLS equation possessing infinitely many conserved densities
is given by the following system  \cite{69, 79, 72}:
\begin{equation}
  \left\{\begin{array}{@{}l@{}}
       p_t = -p_{x x} + 2 \alpha p^2 q_x + 2 \beta p q p_x - \alpha (\beta - 2 \alpha)p^3 q^2\\
  q_t = q_{x x} + 2 \alpha q^2 p_x + 2 \beta p q q_x + \alpha (\beta - 2 \alpha)p^2 q^3
  \end{array}\right.
  \label{system}
\end{equation}
where $p = p(x,t)$ and $q = q(x,t)$ are two functions depending on real variables $x$ and $t$ and $\alpha, \beta \in \mathbb{C}$
are constants. Subscripts are used as shorthand notation for partial derivatives.
Specifying the constants in \eqref{system}, three historically significant cases are identified:
\begin{enumerate}
 \item The Kaup-Newell system \cite{64} arises  when $\alpha = 2$ and $\beta = 4$:
 \begin{equation*}
  \left\{\begin{array}{@{}l@{}}
      p_t = - p_{xx} + 4(p^2  q)_x\\
      q_t=   q_{xx} + 4(p q^2)_x
  \end{array}\right.\tag{A}\label{DNLS0I}
\end{equation*}
 \item The Chen-Lee-Liu system \cite{66} is obtained for $\alpha = 0$ and $\beta = 2$:
  \begin{equation*}
  \left\{\begin{array}{@{}l@{}}
      p_t = - p_{ xx} + 4 p q p_x\\
      q_t =  q_{ xx} + 4 p q q_x
  \end{array}\right.\tag{B}\label{DNLS0II}
\end{equation*}
 \item The Gerdjikov-Ivanov system \cite{115b} corresponds to $\alpha = -2$ and $\beta = 0$:
  \begin{equation*}
  \left\{\begin{array}{@{}l@{}}
      p_t = -p_{xx} - 4 p^2 q_x + 8  p^3 q^2\\
      q_t =  q_{ xx} - 4 q^2 p_x - 8 p^2 q^3
  \end{array}\right.\tag{C}\label{DNLS0III}
\end{equation*}
\end{enumerate}

Note that the DNLS equations \eqref{system} are homogeneous if we assign weight $1$ to the dependent variables $p$ and $q$, while the differentiations in $x$ and $t$ have weights $2$ and $4$ respectively.

Several papers have been dedicated to the study of integrable non-Abelian versions of the equations mentioned above, where the variables $p$ and $q$ take values in a non-commutative associative algebra. In \cite{72}, Olver and Sokolov presented a classification of integrable derivative nonlinear Schrödinger (DNLS) equations. The authors considered systems of the form
\begin{equation}
  \left\{\begin{array}{@{}l@{}}
       p_t = -p_{x x} + F(p,q,p_x,q_x)\\
       q_t = q_{x x} + G(p,q,p_x,q_x)
  \end{array}\right.
  \label{form}
\end{equation}
where $F$ and $G$ are polynomials of weight $5$ with respect to the weights assigned for \eqref{system}. They investigated systems admitting higher-order symmetries of weight $9$ and identified nine non-Abelian equations, up to rescaling of the variables, exchange $p \leftrightarrow q$, and an involution $\star$ (which, for matrix equations, corresponds to the matrix transpose), satisfying
\begin{equation}\label{transpose}
 (f^\star)^\star=f, \qquad (fg)^\star=g^\star f^\star.
\end{equation}

The integrability of these equations was established by Tsuchida and Wadati \cite{74}, who demonstrated that two of the equations in \cite{72} are linearisable, while the remaining seven possess Lax representations. The same set of integrable equations with Lax representations also appeared in \cite{115}, where Adler and Sokolov conducted a broader study of non-Abelian evolutionary systems.

Below, we explicitly list the seven non-Abelian integrable DNLS equations in \cite{72} admitting Lax representations:
\begin{align}
\allowdisplaybreaks
  & \left\{\begin{array}{@{}l@{}}
          p_t =  -  p_{xx} + 4(p   q   p)_x\\
         q_t =  q_{xx} + 4(q   p   q)_x
     \end{array}\right. \tag{A1} \label{DNLSI}\\
    & \left\{\begin{array}{@{}l@{}}
            p_t =-p_{xx} + 4 \left(p^2 q_{x}+p_{x} p q+p_{x} q p\right)+8 \left(p^3 q^2-p^2 q^2 p\right)\\
             q_t =q_{xx}+ 4 \left(p q q_{x}+q p q_{x}+p_{x} q^2\right)+8 \left(q p^2 q^3-p^2 q^3\right)
         \end{array}\right.  \tag{A2}  \label{DNLSI-2}\\
    & \left\{\begin{array}{@{}l@{}}
         p_t = -p_{xx} + 4 p_x q p\\
         q_t = q_{xx} + 4 q p q_x
     \end{array}\right. \tag{B1} \label{DNLSII}\\
    & \left\{\begin{array}{@{}l@{}}
      p_t = - p_{xx} +4\left(p    q   p\right)_x - 4( q_x   p^2 + q   p   p_x )  + 8( q   p^2  q   p - 2 q   p   q   p^2 + q^2   p^3)\\
      q_t = q_{xx}  + 4 \left(q    p   q\right)_x -4( q^2   p_x + q_x   q   p) + 8( 2 q^2   p   q  p  -  q   p   q^2   p - q^3   p^2 )
     \end{array}\right. \tag{B2} \label{DNLSII-2}\\
    & \left\{\begin{array}{@{}l@{}}
      p_t = - p_{xx} + 4 (p^2  q_x  + p_x p  q   - p    q_x   p ) + 8( p^3   q^2  -  p^2   q   p q -  p^2 q^2   p + p  q   p   q   p)\\
      q_t = q_{xx} +  4( p_x q^2 +p q q_x - q   p_x   q )+ 8(  p q   p   q^2 +  q p^2   q^2    -  p^2 q^3 -  q   p   q   p q )
     \end{array}\right. \tag{B3} \label{DNLSII-3}\\
    & \left\{\begin{array}{@{}l@{}}
         p_t = -p_{xx} - 4 p  q_x  p + 8 p  q  p  q  p\\
          q_t =  q_{xx} - 4 q  p_x  q - 8 q  p  q  p  q
     \end{array}\right.\tag{C1}\label{DNLSIII}\\
    & \left\{\begin{array}{@{}l@{}}
            p_t =-p_{xx} +4 (p_x q p- q p p_x- q_x p^2)+8 (q p^2 q p- q p q p^2+ q^2 p^3)\\
             q_t = q_{xx}+4 ( q p q_x- q^2 p_x- q_x q p)+8 (q^2 p q p- q p q^2 p- q^2 p^2)
         \end{array}\right. \tag{C2} \label{DNLSIII-2}
 \end{align}
Here, the labels indicate the connection to the commutative form of the equations. For example, \eqref{DNLS0I} represents the commutative form of both \eqref{DNLSI} and \eqref{DNLSI-2}.  These equations are the primary focus of this paper and are equivalent to those in \cite{72} under the transformations mentioned above.

In this paper, we present the known Lax representations of the above seven DNLS equations by solving a classification problem for specified forms of Lax pairs. We then
discuss their reduction group \cite{101} and gauge transformations. The main results follow in the subsequent sections. In Section \ref{S2}, we introduce fundamental definitions of Darboux transformations and examines their inverse and composition.  Instead of considering generic Darboux matrices for DNLS equations, Section \ref{S6} focuses on Darboux matrices polynomial in the spectral parameter $\lambda \in \mathbb{C}$ that inherit the same reduction group as the Lax representations. Specifically, we consider linear and quadratic rank-1 up Darboux matrices
in the following forms:

\begin{eqnarray}
 &&
 M_\uparrow(1) = \begin{pmatrix}
      f & 0 \\ 0 & 0
     \end{pmatrix} \lambda + \begin{pmatrix}
     0 &  f   p \\
      q_1   f & 0
\end{pmatrix},\label{DarBoux1}\\
 &&M_\uparrow(2)= \begin{pmatrix}
      \hh & 0 \\ 0 & 0
     \end{pmatrix} \lambda^2 + \begin{pmatrix}
     0 &  \hh   p \\  q_1   \hh & 0
     \end{pmatrix} \lambda + \begin{pmatrix}
      a & 0 \\ 0 & b
     \end{pmatrix}.\label{OriginalMatrix}
\end{eqnarray}
The term ``up" (or equivalently ``down" in Section \ref{inverse}) refers to the non-zero eigenvalue in the matrix coefficient for the highest power of $\lambda$.

The quasideterminants of these matrices enable us to compute their inverses and determine whether a polynomial Darboux matrix has the linear Darboux matrix \eqref{DarBoux1} as a factor.
In Section \ref{sec45}, we discuss the factorisation of quadratic Darboux matrices and derive a necessary condition for it.
We show that linear and quadratic rank-1 down-Darboux matrices are related to the inverses of their respective up-Darboux matrices.
Furthermore, the reduction leading to the 2-component Volterra lattice occurs precisely when the quadratic Darboux matrix factorises into two linear Darboux matrices.

In all considered cases of DNLS equations, the linear Darboux transformations with matrices \eqref{DarBoux1} can be cast in either the Volterra or the modified Volterra equation:
    \begin{align}
        &  u_x=   2 (\mu_{1} u_{1} u - u u_{-1}\mu_{-1}), \label{Va}\tag{V$_a$}\\
        & u_x=   2 u (  u_1 \mu_1 -  \mu u_{-1})u, \label{mVb} \tag{mV$_b$}
    \end{align}
where $\mu$ satisfies $\mu_x = 0$ and it is a non-commutative constant with respect to $x$.
Note that there is no direct Miura transformation between them, as discussed in \cite{112} and in Appendix \ref{Appb-V}.
The reductions of the quadratic Darboux transformations with matrix \eqref{OriginalMatrix} lead to the following two-component
integrable D$\Delta$Es:
\allowdisplaybreaks
    \begin{align}
    & \left\{\begin{array}{@{}l@{}}
   u_{x}= 2 \mu_1 u_1( \nu  - v u)\\
   v_{x}= 2 (v u  -  \nu) v_{-1} \mu
  \end{array}\right.\label{AL}\tag{AL}\\
  & \left\{\begin{array}{@{}l@{}}
         u_x = 2(\mu u_1\nu - u v u)\\
         v_x = 2(- \nu_{-1} v_{-1}\mu_{-1} + v u v )
    \end{array}\right.\label{MRT}\tag{MRT}\\
        &\left\{\begin{array}{@{}l@{}}
      u_x=   2  v_1 (u_1+\mu_1) u - 2 u v (u + \mu) \\
      v_x = 2 (v u v -  u_{-1} v_{-1} v)
  \end{array}\right. \label{N1} \tag{N$_1$}\\
        &   \left\{\begin{array}{@{}l@{}}
     u_x = 2 \left(u   \mu   v -  v_1  \mu_1  u \right) + 2  (v_1   u_1   v_1^{-1} u -  u   v^{-1} u_{-1}  v)\\
      v_x = 2(  v u -u_{-1}   v  )
  \end{array}\right. \label{N2} \tag{N$_2$}\\
  & \left\{\begin{array}{@{}l@{}}
     u_x = 2(u v_1 - v u + u u_1 \mu_1 - u \mu u)\\
     v_x = 2(v u \mu - \mu_{-1} v_{-1} u_{-1})
    \end{array}\right.\label{N3}\tag{N$_3$}\\
    & \left\{\begin{array}{@{}l@{}}
    u_{x}= 2 ( uv\mu - \mu_1v_1u) +2 (u\mu^{-1} u_{-1}\mu -\mu_1u_1\mu^{-1}_1u)\\
   v_{x}= 2(v u_{-1} - u v)
 \end{array}\right.\label{rT}\tag{rT}\\
 & \left\{\begin{array}{@{}l@{}}
   u_{x}= 2(\mu_1 v_1 u - u v \mu)\\
   v_{x}= 2(u v - v u_{-1})
  \end{array}\right.\label{2Va}\tag{2V$_a$}\\
    &\left\{\begin{array}{@{}l@{}}
   u_x = 2( v_1 \mu_1   u - u \mu v) \\
      v_x = 2( v  u -  u_{-1} v )
  \end{array}\right. \label{2Vb} \tag{2V$_b$}
    \end{align}
Here $\mu$ and $\nu$ are non-commutative constants satisfying $\mu_x=\nu_x=0$. The resulting D$\Delta$Es are considered up to rescaling of the variables\footnote{To maintain the same scaling as in the DNLS systems, we retain the factor $2$ rather than rescaling it to $1$.}, an involution $\star$ defined by \eqref{transpose} and another involution $\cI$ (reflection operator) satisfying
\begin{equation}\label{involp}
  \cI(p_{-n})=p_n, \quad\cI(q_{-n})=q_n, \quad \cI(f g)=\cI(f) \cI(g),
\end{equation}
where $f$ and $g$ are functions of $p$, $q$ and their shifts.\par

In the list above, several systems are generalisations of known non-Abelian integrable systems containing non-commutative constants. System \eqref{AL} corresponds to the Ablowitz-Ladik system \cite{178,177} and its non-Abelian version is given in \cite{26}. System \eqref{MRT} denotes the Merola-Ragnisco-Tu system \cite{163,162}. System \eqref{rT} represents the relativistic Toda equation \cite{174,175,176} and its non-Abelian version is introduced in \cite{164,165}. Systems \eqref{2Va} and \eqref{2Vb} are two versions of the 2-component Volterra lattice.
Notably, system \eqref{N2} shares the same commutative form (under the rescaling $u\mapsto -u$) as \eqref{rT}. It can therefore be regarded
as a second non-Abelian lift of the relativistic Toda lattice. To the best of our knowledge, systems \eqref{N1}, \eqref{N2} and \eqref{N3}  are novel and have not previously appeared in the literature.

We conclude the paper with  brief remarks on the implications and potential directions for future studies. We also include
an appendix exploring Miura transformations among the obtained integrable D$\Delta$Es.
Transformations among the Volterra-related lattices including systems \eqref{2Va} and \eqref{2Vb} are discussed in Section \ref{Appb-V}.
In Section \ref{Appb-Q1}, we show that
system \eqref{AL} can be transformed into systems \eqref{N1} and \eqref{N2} via Miura transformations. Finally, we derive Miura transformations connecting the following generalisation of the Kaup lattice \cite{133,26} with a non-commutative constant $\mu$:
\begin{equation}
    \left\{\begin{array}{@{}l@{}}
         u_x = 2(\mu_{1}u_{1} - u \mu)(u + v)\\
         v_x = 2(u + v)(\mu v - v_{-1} \mu_{-1})
    \end{array}\right.
      \label{K}\tag{K}
\end{equation}
and system \eqref{MRT} to system \eqref{N3} in Section \ref{Appb-Q2}.

\section{Lax representations and gauge transformations}\label{Sn3}
Given an evolutionary PDE considered to be S-integrable \cite{134}, one associates to it a pair of linear operators $D_x-U$ and $D_t-V$,  defining the following auxiliary system
\begin{equation}
 \left\{\begin{array}{@{}l@{}}
     \Phi_x = U(p,q;\lambda) \Phi\\
     \Phi_t = V(p,q;\lambda) \Phi
 \end{array}\right.
 \label{L1}
\end{equation}
where $\Phi$ is a fundamental solution. Here $U$ and $V$ are square matrices whose entries depend on
the variables $p, q$, their $x$-derivatives, and certain rational (or, in some cases, elliptic) functions of
the spectral parameter $\lambda \in \mathbb{C}$. The compatibility condition of the auxiliary system, given by
\begin{equation}
 U_t- V_x +[U,V] = 0,
 \label{L2}
 \end{equation}
 where the bracket $[\cdot,\ \cdot]$ denotes the matrix commutator, is equivalent to the original evolutionary equation. Equation \eqref{L2} is also referred to as the zero curvature representation,
 and the pair $(U,V)$ is conventionally called a Lax representation of the original evolutionary equation.

Given an invertible matrix $G$, we perform a gauge transformation of the fundamental solution as $\Phi=G \Psi$. The new auxiliary system \eqref{L1} in $\Psi$ becomes
\begin{equation*}
 \left\{\begin{array}{@{}l@{}}
    \Psi_x = U'\Psi\\
     \Psi_t = V' \Psi
 \end{array}\right.
\end{equation*}
and it is associated with a new Lax representation
\begin{equation}
 U \mapsto  U' = G^{-1} U G - G^{-1}G_x, \qquad
    V \mapsto  V' = G^{-1} V G - G^{-1}G_t.
 \label{Gauge}
\end{equation}
Under this transformation, the zero curvature condition (and the equivalent evolutionary PDE) remains invariant \cite{11}. In this paper, especially in Section \ref{S6}, we employ gauge transformations to refine the obtained Lax representations and eliminate specific dependent variables (interpreted as gauge degrees of freedom).

\subsection{Lax representation of the DNLS equations}\label{S3}
In this section, we conduct a small-scale classification study for the Lax representations of the DNLS equations. Given two $2 \times 2$ matrices $U$ and $V$ of a specific form associated with DNLS equations, we classify all
cases that satisfy the zero curvature condition (\ref{L2}). Building on this classification, we present the Lax representations for the DNLS equations \eqref{DNLSI}-\eqref{DNLSIII-2}.

These representations were systematically presented by Tsuchida and Wadati in \cite{74}, although some of the equations had alternative Lax representations known prior to the cited article. Different but equivalent Lax representations for the same models can also be found in \cite{115}.

Consider $2 \times 2$ matrices $U$ and $V$ in the following form:
\begin{subequations}
\begin{align}
    U = & 
    {I}\lambda^2 + 2
    J \lambda+ 2 P, \label{LaxP1-A} \\
     V = & -2I\lambda^4 - 4 J  \lambda^3 + 4 IJJ \lambda^2  + 2 Q \lambda - 2 R, \label{LaxP1-B}
\end{align}
\label{LaxP1}
\end{subequations}
where the matrices $I$ and $J$ are given by
\begin{equation}
 I = \begin{pmatrix}
                 1 & 0 \\ 0 & -1
                \end{pmatrix}, \qquad
 J = \begin{pmatrix}
                 0 & p \\ q & 0
                \end{pmatrix} .
    \label{mat}
\end{equation}
Here, $P$ and $R$ are diagonal matrices, while $Q$ has only off-diagonal non-zero entries. All their entries are functions of $p$, $q$ and their $x$-derivatives. Note that $I$ corresponds to the third Pauli matrix.

Assuming that $U$ and $V$ in \eqref{LaxP1} form a Lax representation for a PDE, we deduce the constraints among the undetermined matrices $P, Q$ and $R$ such that they satisfy \eqref{L2}. These conditions allow us to explore the most general dynamics associated with a Lax representation of the form \eqref{LaxP1}.
\begin{prop}\label{prop1}
The matrices $U$ and $V$ defined by \eqref{LaxP1} form a Lax representation for the PDE
\begin{equation}
 J_t + I J_{xx} + 2 I[J,P_x] + 4 I[J_x,P] - 4 (J^3)_x + 8 [P,J^3] + 4 I[P,[P,J]] - 2[J,R] = 0
 \label{mostgenDNLS}
\end{equation} 
if and only if $Q$ and $R$ in \eqref{LaxP1-B} satisfy the conditions
 \begin{eqnarray}
 &&Q = 2I [P,J] + 4 J^3 -  I J_x,
 \label{Q}\\
&& P_t + R_x = 2[P,R] .
 \label{cond}
\end{eqnarray}
\end{prop}
\begin{proof}
Recall that $P$ and $R$ are diagonal matrices, while $Q$ has off-diagonal non-zero entries. The following commutation and anti-commutation relations hold:
 \begin{align*}
   [I, A] = I A - A I = 0, & \qquad  \text{if $A$ is a diagonal matrix };\\
   \{I, B\} = I B + B I = 0, & \qquad \text{if $B$ is an off-diagonal matrix}.
 \end{align*}
Additionally, it is obvious that $I_t = I_x = 0$ and that $I^2 = \mathds{1}$, the identity matrix.

The statement is then proven by substituting \eqref{LaxP1} into \eqref{L2} and collecting the coefficients of the resulting polynomial in $\lambda$. Using the properties above, we obtain the following equations:
\begin{subequations}
 \begin{align}
  & \lambda^3: \qquad  J_{x}- 4 I J J J+ 2 J P- 2 P J+ I Q = 0; \label{syst3}\\
  & \lambda^2: \qquad I J J_{x}+I J_{x} J + 2 I J J P-2 I P J J- J Q + Q J = 0; \label{syst2}\\
  & \lambda^1: \qquad J_{t}-Q_{x}-2 J R+2 R J+2 P Q-2 Q P = 0; \label{syst1}\\
  & \lambda^0: \qquad  P_{t}+R_{x}-2 P R+2 R P = 0. \label{syst0}
 \end{align}
\end{subequations}
Multiplying \eqref{syst3} by $I$ on the left, we recover $Q$ as it appears in \eqref{Q}.
For such $Q$, \eqref{syst2} is also satisfied identically. Moreover, the $x$-derivative of $Q$ is
\begin{equation*}
 Q_x = -I J_{x x} +2 I \left( P J_{x}+ P_{x} J - J P_{x}- J_{x} P\right)+4 \left(J J J_{x}+J J_{x} J+J_{x} J J\right).
\end{equation*}
Substituting this into \eqref{syst1} yields equation \eqref{mostgenDNLS}. Finally, the coefficient \eqref{syst0} corresponds to the condition \eqref{cond}.
\end{proof}
\begin{remark}
The equivalent version of Proposition \ref{prop1} for the commutative DNLS equations is analogous. We can consider the same statement, with the additional constraint that diagonal matrices commute.
Such condition significantly simplifies \eqref{cond}, which reduces to:
\begin{equation*}
 P_t + R_x  = 0.
\end{equation*}
\end{remark}

It follows from Proposition \ref{prop1} that we need to determine the diagonal matrices $P={\rm diag}(P_{ii})_{2\times 2}$ and $R={\rm diag}(R_{ii})_{2\times 2}$ to obtain Lax representations for (\ref{mostgenDNLS}). Due to the homogeneity of DNLS equations, we take $P_{ii}$ and $R_{ii}$, $i=1,2$, to be homogeneous polynomials of $p$, $q$ and their $x$-derivatives with weights $2$ and $4$ respectively. Thus we let
\begin{subequations}
    \begin{align}
        & P_{ii} = \al_i^{(1)} p^2 + \al_i^{(2)} p q + \al_i^{(3)} q p + \al_i^{(4)} q^2,\\
        & R_{ii} = \be_i^{(1)} p^4+ \be_i^{(2)}p^3q+ \be_i^{(3)} p^2qp+ \be_i^{(4)} pqp^2+ \be_i^{(5)} qp^3+ \be_i^{(6)}p^2q^2 + \notag  \\ &\qquad + \be_i^{(7)} pqpq+ \be_i^{(8)} pq^2p+ \be_i^{(9)} q p^2 q+ \be_i^{(10)} qpqp+ \be_i^{(11)} q^2p^2 + \notag \\ &\qquad+ \be_i^{(12)} q^3p+ \be_i^{(13)} q^2pq+ \be_i^{(14)} qpq^2+ \be_i^{(15)} pq^3+ \be_i^{(16)} q^4 + \\ &\qquad+ \be_i^{(17)}p_x p+ \be_i^{(18)} p p_x+ \be_i^{(19)} p_x q+ \be_i^{(20)} p q_x + \be_i^{(21)} q_x p + \notag  \\ &\qquad +\be_i^{(22)} q p_x +  \be_i^{(23)}q_x q + \be_i^{(24)}q q_x, \notag
    \end{align}
    \label{ans}
\end{subequations}
where $\al_i^{(j)}, \be_i^{(l)}\in \mathbb{C}$ are constants with $i=1,2$, $j=1,\dots, 4$ and $l=1,\dots, 24$. We determine their
values using Proposition \ref{prop1}.
The condition \eqref{cond} reduces most of the parameters in \eqref{ans} to zeros. The remaining non-zero parameters are $\al_1^{(2)}, \al_1^{(3)}, \al_2^{(2)}$ and $\al_2^{(3)}$, determined by the following algebraic system:
\begin{equation}
    \left\{\begin{array}{@{}l@{}}
         \al_1^{(2)}(\al_1^{(2)}+1)=0\\
         \al_2^{(3)}(\al_2^{(3)}-1)=0\\
         \al_1^{(3)}\al_2^{(2)}=0\\
         \al_2^{(2)}(1+2 \al_1^{(2)} - \al_2^{(2)}) = 0\\
         \al_1^{(3)} (1+\al_1^{(3)}-2 \al_2^{(3)}) =  0
    \end{array}\right.
    \label{Equa}
\end{equation}
The expressions in \eqref{ans} become
    \begin{align*}
        & P_{11} = \al_1^{(2)} p q + \al_1^{(3)} q p;\\
        & P_{22} = \al_2^{(2)} p q + \al_2^{(3)} q p; \\
        & R_{11} =\left(4 (\al_2^{(3)}-1) \al_1^{(2)}-2 (\al_1^{(2)})^2\right) p q p q+2 \al_1^{(3)} (2 \al_2^{(2)}-\al_1^{(2)}) q p^2 q+\\  & \qquad +\al_1^{(3)} q p_x-\al_1^{(3)} q_x p-2 \al_1^{(2)} \al_1^{(3)} p q^2 p-2 \al_1^{(3)} (2 \al_1^{(2)}-\al_1^{(3)}+2) q p q p +\notag\\  & \qquad+4 \al_1^{(2)} \al_2^{(2)} p^2 q^2-\al_1^{(2)} p q_x+\al_1^{(2)} p_x q-4 (\al_1^{(3)})^2 q^2 p^2;\notag\\
        & R_{22} = \al_2^{(3)} q p_x-\al_2^{(3)} q_x p-2 \al_2^{(3)} (-\al_2^{(3)}+2 \al_1^{(2)}+2) q p q p +\\  & \qquad +p q^2 p (2 \al_2^{(3)} \al_2^{(2)}-4 \al_1^{(3)} \al_2^{(2)})-4 \al_2^{(3)} \al_1^{(3)} q^2 p^2 +4 (\al_2^{(2)})^2 p^2 q^2 + \notag\\  & \qquad-2 \al_2^{(2)} (-2 \al_2^{(3)}+\al_2^{(2)}+2) p q p q+2 \al_2^{(3)} \al_2^{(2)} q p^2 q+\al_2^{(2)} p_x q  -\al_2^{(2)} p q_x.\notag
    \end{align*}
Solving system \eqref{Equa} and determining the associated DNLS equations, we obtain the following statement:
\begin{thm}\label{thm1}
 Up to rescaling of the variables, exchange $p \leftrightarrow q$, and the involution $\star$ defined in \eqref{transpose},
 the DNLS equations \eqref{DNLSI}-\eqref{DNLSIII-2} are the only seven systems admitting a Lax representation with the matrices $U$ and $V$ of the form \eqref{LaxP1} and $P$, $ R$ given by \eqref{ans}.
\end{thm}
\begin{proof}
Directly solving system \eqref{Equa}, we obtain 12 cases:
    \begin{align*}
        & ({\rm i}).\quad \al_1^{(2)}=\al_1^{(3)}=\al_2^{(2)}=\al_2^{(3)}=0; && ({\rm ii}).\quad \al_1^{(2)}=\al_1^{(3)}=\al_2^{(2)}=0,\  \al_2^{(3)}=1; \\
        & ({\rm iii}).\quad \al_1^{(2)}=\al_1^{(3)}=\al_2^{(3)}=0,\ \al_2^{(2)}=1; && ({\rm iv}).\quad \al_1^{(2)}=\al_2^{(2)}=\al_2^{(3)}=0, \ \al_1^{(3)}=-1; \\
       & ({\rm v}).\quad \al_1^{(3)}=\al_2^{(2)}=\al_2^{(3)}=0,\ \al_1^{(2)}=-1; && ({\rm vi}).\quad \al_1^{(2)}=\al_2^{(2)}=0,\ \al_1^{(3)}=\al_2^{(3)}=1; \\
        & ({\rm vii}).\quad \al_1^{(2)}=\al_1^{(3)}=0,\ \al_2^{(2)}=\al_2^{(3)}=1; && ({\rm viii}).\quad \al_1^{(2)}=\al_1^{(3)}=-1,\ \al_2^{(2)}=\al_2^{(3)}=0; \\
        & ({\rm ix}).\quad \al_1^{(2)}=\al_2^{(2)}=-1,\ \al_1^{(3)}=\al_2^{(3)}=0; && ({\rm x}).\quad \al_1^{(2)}=-1,\ \al_1^{(3)}=\al_2^{(2)}=0,\ \al_2^{(3)}=1; \\
        & ({\rm xi}).\quad \al_1^{(2)}=-1,\ \al_2^{(2)}=0,\ \al_1^{(3)}=\al_2^{(3)}=1; && ({\rm xii}).\quad \al_1^{(2)}=\al_2^{(2)}=-1,\ \al_1^{(3)}=0,\ \al_2^{(3)}=1.
    \end{align*}
We now derive system \eqref{mostgenDNLS} case by case. Direct computation shows that cases
(i), (vi), (v), (iv) (xii), (x) and (viii) yield the DNLS equations \eqref{DNLSI}-\eqref{DNLSIII-2}, respectively.
Case (iii) corresponds to equation
\begin{equation*}
 \left\{\begin{array}{@{}l@{}}
             p_t =-p_{xx} +4 (p q p)_x -4 \left( p^2 q_x+ p_x p q \right)+8 \left( p^3 q^2-2 p^2 q p q+ p q p^2 q\right)\\
             q_t =\ \ q_{xx} +4 (q p q)_x -4\left( p q q_x+ p_x q^2\right)-8 \left( p^2 q^3 -2 p q p q^2+ p q^2 p q\right)
        \end{array}\right.
\end{equation*}
which is equivalent to equation \eqref{DNLSII-2} under the involution $\star$. Similarly, cases (ii), (vii), (ix) and (xi) are related to equations \eqref{DNLSII}, \eqref{DNLSIII-2}, \eqref{DNLSI-2} and \eqref{DNLSII-3}, respectively, via suitable transformations. This completes the proof.
\end{proof}
Following the above theorem, we list the Lax representations of the seven DNLS equations \eqref{DNLSI}-\eqref{DNLSIII-2}, which we are going to use subsequently. Most are found in \cite{74}. We explicitly provide only the expressions for $P$ and $R$ as $Q$ can be derived from \eqref{Q}.\\
The Lax representation of \eqref{DNLSI} is
\begin{equation*}
 \begin{gathered}
  P =  
    \begin{pmatrix}
            0 & 0 \\
            0 & 0\\
    \end{pmatrix},   \qquad
    R = \begin{pmatrix}
            0 & 0\\
            0 &  0
    \end{pmatrix} .
 \end{gathered}
 \tag{Lax-A1}
\label{LaxA1}
\end{equation*}
The Lax representation of \eqref{DNLSI-2} is
\begin{equation*}
 \begin{gathered}
  P =  
    \begin{pmatrix}
            - p q & 0 \\
            0 & - p q\\
    \end{pmatrix},   \qquad
    R = \begin{pmatrix}
             p q_{x}-p_{x} q+4 p^2 q^2+2 (p q)^2 & 0 \\ 0 & p q_{x}-p_{x} q+4 p^2 q^2+2 (p q)^2
        \end{pmatrix}.
 \end{gathered}
 \tag{Lax-A2}
\label{LaxA2}
\end{equation*}
The Lax representation of \eqref{DNLSII} is
\begin{equation*}
 \begin{gathered}
  P =  
    \begin{pmatrix}
             -p q & 0 \\
            0 & 0\\
    \end{pmatrix},  \qquad
    R = \begin{pmatrix}
            p q_x - p_x q  +2 (p q)^2 & 0 \\ 0 & 0
        \end{pmatrix}.
 \end{gathered}
 \tag{Lax-B1}
\label{LaxB1}
\end{equation*}
The Lax representation of \eqref{DNLSII-2} is
\begin{equation*}
 \begin{gathered}
  P =  
    \begin{pmatrix}
            -q p & 0 \\
            0 & 0\\
    \end{pmatrix},  \qquad
    R = \begin{pmatrix}
            q_x p-q p_x + 6 (q p)^2 - 4 q^2 p^2 & 0\\
            0 & 0
    \end{pmatrix}.
 \end{gathered}
 \tag{Lax-B2}
\label{LaxB2}
\end{equation*}
The Lax representation of \eqref{DNLSII-3} is
\begin{equation*}
 \begin{gathered}
  P =  
    \begin{pmatrix}
            -p q & 0 \\
            0 & q p - p q\\
    \end{pmatrix}, \qquad
    \\
    R = \begin{pmatrix}
             p q_x - p_x q+ 4 p^2 q^2 - 2(p q)^2 & 0\\ 0 & \begin{array}{l}p q_x + q p_x - p_x q - q_x p + 2 (q p)^2\\ - 2 (p q)^2 -2 p q^2 p - 2 q p^2 q + 4 p^2 q^2\end{array}
    \end{pmatrix} .
 \end{gathered}
 \tag{Lax-B3}
\label{LaxB3}
\end{equation*}
The Lax representation of \eqref{DNLSIII} is
\begin{equation*}
 \begin{gathered}
  P =  
    \begin{pmatrix}
            -p q & 0 \\
            0 & q p\\
    \end{pmatrix}, \qquad
    R = \begin{pmatrix}
             p q_x - p_x q - 2(p q)^2 & 0\\ 0 &  q p_x - q_x p + 2 (q p)^2
    \end{pmatrix}.
 \end{gathered}
 \tag{Lax-C1}
\label{LaxC1}
\end{equation*}
Finally, The Lax representation of \eqref{DNLSIII-2} is
\begin{equation*}
 \begin{gathered}
  P =  
    \begin{pmatrix}
            -q p - p q & 0 \\
            0 & 0\\
    \end{pmatrix},   \qquad
    R = \begin{pmatrix}
             \begin{array}{l}p q_x-q p_x-p_x q+q_x p-2 p q^2 p\\-2 q p^2 q + 2 (p q)^2 + 2 (q p)^2- 4 q^2 p^2\end{array} & 0 \\ 0 & 0
        \end{pmatrix}.
 \end{gathered}
 \tag{Lax-C2}
\label{LaxC2}
\end{equation*}
The reduction group is a powerful tool to encapsulate the symmetry properties of Lax representations. This has proven useful both for classifying known Lax pairs and discovering new ones.  Originally introduced in \cite{138b,31,101}, the concept of reduction group
has yielded numerous insights, including those in \cite{78,139}.
The generic Lax representation $(U,V)$ presented in \eqref{LaxP1} is invariant under the following action:
\begin{equation}
 U(\lambda) \mapsto I U(-\lambda) I, \qquad V(\lambda) \mapsto I V(-\lambda) I, \label{reduction1}
\end{equation}
where $I$ is defined in \eqref{mat}.
The above group action is isomorphic to $\mathbb{Z}_2$, and it defines a reduction group for all DNLS Lax representations.

\subsection{Gauge transformations for DNLS equations}\label{S5}
The concept of a gauge transformation for non-Abelian DNLS equations was introduced by Tsuchida and Wadati in \cite{74} to prove the integrability of their matrix valued versions.
In this paper, we adapt this gauge transformation technique to simplify Darboux matrices.

Consider a gauge transformation \eqref{Gauge} defined by the matrix
\begin{equation}
 G = \begin{pmatrix}
      A & 0 \\ 0 & B
     \end{pmatrix},
     \label{key}
\end{equation}
where $A$ and $B$ are non-zero functions depending on $p$ and $q$. Following \eqref{Gauge}, the Lax representation \eqref{LaxP1} becomes
\begin{subequations}
\begin{align}
 U' & =  
    I \lambda^2 + 2
    \begin{pmatrix}
            0 & A^{-1} p B \\
            B^{-1} q  A & 0\\
    \end{pmatrix} \lambda+ 2  G^{-1} P G - G^{-1} G_x,\\
    V' & = -2I \lambda^4 -4 \begin{pmatrix}
            0 &  A^{-1} p B \\
            B^{-1} q A & 0\\
    \end{pmatrix}  \lambda^3  + 4 \begin{pmatrix}
            A^{-1} p q  A & 0 \\
            0 & - B^{-1} q  p B \\
    \end{pmatrix} \lambda^2 + \\ & \qquad  +2 G^{-1} Q G  \lambda -2 G^{-1} R G -G^{-1} G_t.\notag
    \end{align}
    \label{30}
\end{subequations}
If we introduce a transformation of the dependent variables:
\begin{equation}
 p \mapsto p' = A^{-1} p B , \qquad q \mapsto q' = B^{-1} q A,
 \label{para}
\end{equation}
we can express the gauged pair $(U', V')$ in a form analogous to \eqref{LaxP1}:
\begin{subequations}
\begin{align}
 U' & =  
    I \lambda^2 + 2
    J'  \lambda+ 2 P',\\
    V' & = -2 I \lambda^4 - 4  J' \lambda^3 + 4 I J' J' \lambda^2 + 2 Q' \lambda - 2   R'.
    \end{align}
    \label{LaxP2}
\end{subequations}
Here, $J' = J(p',q')$ is defined similarly to $J$ in \eqref{mat}. The matrices $P'$ and $R'$ are:
\begin{subequations}
 \begin{align}
  & P' = \begin{pmatrix}
       A^{-1} \tilde P_{11} A - \frac{1}{2} A^{-1} A_x & 0 \\ 0 & B^{-1} \tilde P_{22} B - \frac{1}{2} B^{-1} B_x
      \end{pmatrix},\\
      & R' = \begin{pmatrix}
       A^{-1} \tilde R_{11} A + \frac{1}{2} A^{-1} A_t & 0 \\ 0 & B^{-1} \tilde R_{22} B + \frac{1}{2} B^{-1} B_t
      \end{pmatrix},
 \end{align}
 \label{ultima}
\end{subequations}
where we present a function $f(p,q)$ with a tilde as $\tilde{f}=\tilde{f}(p',q')= f(A p' B^{-1}, B q' A^{-1})$. The transformed matrix $Q' = G^{-1} \tilde Q G$ satisfies \eqref{Q} for $P'$ and $R'$.

In matrix notation, this gauge transformation simplifies to:
\begin{equation*}
 J \mapsto J' = G^{-1} \tilde J G, \qquad P \mapsto P' = G^{-1} \tilde P G - \frac{1}{2} G^{-1} G_x, \qquad 
 R \mapsto R' = G^{-1} \tilde R G + \frac{1}{2} G^{-1} G_t.
\end{equation*}
For \eqref{LaxP2} to be a Lax representation of a DNLS equation, every entry of $P'$ and $R'$ must be expressed in closed polynomial form in $p', q'$ and their $x$-derivatives. This ensures that the gauge transformation maps Lax pairs of the form \eqref{LaxP1} to Lax pairs of the same form.

Given the Lax matrices $U$ \eqref{LaxP1} and $U'$ \eqref{LaxP2} defined by $P_{11}$, $P_{22}$ and $P'_{11}$, $P'_{22}$ respectively, both representing DNLS equations, a gauge transformation given by \eqref{key} and \eqref{para} is determined by
\begin{subequations}
    \begin{align}
        & A^{-1} A_{x} = 2(A^{-1} \tilde{P}_{11} A- P'_{11}); && B^{-1} B_{x}  = 2( B^{-1}\tilde{P}_{22} B- P'_{22});
 \label{gaugeaffect1}\\
& A^{-1} A_{t} = 2(R'_{11} - A^{-1} \tilde{R}_{11} A); && B^{-1} B_{t}  = 2( R'_{22} - B^{-1}\tilde{R}_{22} B); \label{gaugeaffect2}
    \end{align}
\end{subequations}
provided $A^{-1} \tilde{P}_{11} A$, $B^{-1}\tilde{P}_{22} B$, $A^{-1} \tilde{R}_{11} A$ and $B^{-1}\tilde{R}_{22} B$ can be expressed in the terms of $p', q'$ and their $x$-derivatives.
This holds trivially for \eqref{LaxA1} (where $P=0$), establishing a gauge transformation from \eqref{DNLSI} to all other DNLS equations.
In \cite{74}, the authors used the gauge transformations from the Lax representation of \eqref{LaxB1} to construct all other Lax representations, proving the integrability of the DNLS equations.

Before we end this section, we recover the known gauge transformation \cite{69} among the commutative DNLS equations \eqref{system} using \eqref{gaugeaffect1}. In the commutative setting, equation \eqref{para} reduces to
 \begin{equation*}
  p \mapsto p' =  E p, \qquad q \mapsto q' = E^{-1} q, \qquad E = A^{-1} B.
 \end{equation*}

From \eqref{gaugeaffect1}, the transformation relating $P$ and $P'$ satisfies:
\begin{equation*}
E_x E^{-1} = 2(\tilde{P}_{22} - \tilde{P}_{11} + P_{11}' - P_{22}').
\end{equation*}
Solving this ordinary differential equation suggests $E = e^{-w}$ for a function $w$. The commutative versions of the systems \eqref{DNLSI}-\eqref{DNLSIII-2} are \eqref{DNLS0I}-\eqref{DNLS0III}, whose Lax representations are \eqref{LaxA1}, \eqref{LaxB1}, \eqref{LaxC1} with commutative multiplication. In these cases, we have
\begin{equation*}
w_x = a p q, \qquad a \in \mathbb{C}.
\end{equation*}
Similarly, we deduce from \eqref{gaugeaffect2} that
\begin{equation*}
 w_t = a \left(p q_x - p_x q + \left(\alpha + \beta\right)p^2 q^2\right), \qquad a \in \mathbb{C}.
\end{equation*}
These two partial derivatives are compatible in the sense that $(w_x)_t = (w_t)_x$ for solutions of the original system \eqref{system} with parameter $\alpha, \beta$.
This mapping of one DNLS equation to another is equivalent to reparametrising $\alpha, \beta$ as $\alpha', \beta'$ in \eqref{system}.
Thus, we define the gauge transformation in terms of $w$ as
\begin{equation*}
 p \mapsto p' = e^{-w} p, \qquad q \mapsto q' =  e^{w} q; \qquad \alpha \mapsto \alpha' = \alpha - a, \qquad \beta\mapsto\beta'=\beta-a.
\end{equation*}
For instance, the Kaup-Newell system \eqref{DNLS0I} transforms into the Chen-Lee-Liu \eqref{DNLS0II} when $a = 2$ and into the Gerdjikov-Ivanov \eqref{DNLS0III} when $a = 4$.

\section{Darboux transformations}\label{S2}
Given an auxiliary system as \eqref{L1} and a fundamental solution $\Phi$, a Darboux transformation $\mathcal{S}$ is an invertible linear transformation
 \begin{equation}
  \mathcal{S}: \Phi  \mapsto \bar{\Phi} = M \Phi,\qquad \det(M)\ne 0,
  \label{D1}
 \end{equation}
mapping $\Phi$ to a fundamental solution $\bar\Phi$ of a new auxiliary system
\begin{equation*}\label{lax1}
 \left\{\begin{array}{@{}l@{}}
     \bar\Phi_x = U(\bar p,\bar q;\lambda) \bar\Phi\\
     \bar\Phi_t = V(\bar p,\bar q;\lambda) \bar\Phi
 \end{array}\right.
\end{equation*}
where the new potentials  $\overline{p}$ and $\overline{q}$ still satisfy the same original PDE as $p$ and  $q$.
The matrix $M = M(\bar{p}, \bar{q},p,q;\lambda)$ that represents the linear transformation \eqref{D1}, is known as Darboux matrix.
Its entries are functions of $p$, $q$, $\bar{p}$, $\bar{q}$ and the spectral parameter $\lambda$. Additional auxiliary functions or parameters may also be involved in $M$.

The map \eqref{D1} is invertible ($\det(M)\ne 0$) and it can be iterated in both directions
\[
  \cdots\munderbar{\Phi}\acS \Phi\acS\bar{\Phi}
  \acS\bar{\bar{\Phi}}\acS\cdots\, .
\]
We introduce the notation
    \begin{align*}
        & \ldots && \Phi_{-1}=\munderbar{\Phi},&& \Phi_0=\Phi, && \Phi_1=\bar{\Phi}, &&\Phi_2=\bar{\bar{\Phi}}, && \dots\\
        & \ldots && {p}_{-1}=\munderbar{p}, && {p}_0={p}, && {p}_1=\bar{p}, && {p}_2=\bar{\bar{p}}, && \dots\\
        & \ldots && {q}_{-1}=\munderbar{q}, &&  {q}_0={q}, && {q}_1=\bar{q}, &&   {q}_2=\bar{\bar{q}}, && \dots
    \end{align*}
The map $\cS$ increases and decreases the subscript index by one, playing the role of shift operator on a $\mathbb{Z}$ lattice. When the subscript index is zero, we often omit it.
The Darboux transformation $\mathcal{S}$ acts on a function $a$ of the variables $p, q$ and their shifts as
$$
\mathcal{S}:a(p_n, q_n, \cdots, p_m, q_m)\mapsto a_1 =  a(p_{n+1}, q_{n+1}, \cdots, p_{m+1}, q_{m+1}),\quad \mathcal{S}:\alpha\mapsto \alpha; \qquad n,m \in \mathbb{Z},\quad  \alpha\in \mathbb{C}.
$$
Similarly, we define the action of the shift operator on matrix $U$ as $\cS(U) = U_1= U(p_1, q_1; \lambda)$, which is the same matrix $U$ with shifted entries. Iterating the transformation once more yields $p_2, q_2$ and so on.

Combining the definition of Darboux transformation \eqref{D1} and the auxiliary system \eqref{L1} of a PDE, we obtain an over-determined auxiliary system
\begin{equation*}
    \left\{\begin{array}{@{}l@{}}
         \cS(\Phi) = M \Phi\\
        \Phi_x = U \Phi
    \end{array}\right.
\end{equation*}
Their compatibility condition leads to the zero curvature equation
\begin{equation}
    M_x = U_1 M -M U.
  \label{D3}
\end{equation}
The resulting equation relates two different solutions of the PDE and is often called an (auto)-B\"acklund transformation.
Equation \eqref{D3} can be used
to determine the Darboux matrix $M$ and provides the Lax representation for the resulting integrable D$\Delta$E.

Similarly to \eqref{Gauge}, a gauge transformation leaves equation (\ref{D3}) invariant. For an invertible matrix $G$ and a generic function $\zeta(\lambda)$ of the spectral parameter, the gauge invariance of \eqref{D3} is given by
\begin{equation}
   M  \mapsto  M' = \zeta(\lambda) G_1^{-1} M G.
\label{gauge}
\end{equation}
In this paper, we primarily consider diagonal matrices $G$.

\subsection{Inverse and composition}\label{secinco}
The inverse map $\mathcal{S}^{-1}$ of a Darboux transformation $\cS$ is the linear transformation
\begin{equation}
\mathcal{S}^{-1} : \Phi \mapsto \Phi_{-1} = M^{I} \Phi,
\label{ID1}
\end{equation}
where $ M^I = (M_{-1})^{-1} $ denotes the Darboux inverse matrix.
Note that the Darboux inverse matrix is the Darboux matrix related to the inverse of the Darboux transformation, and not simply the inverse of the Darboux matrix. The auxiliary system for $M^I$ is defined by
\begin{equation*}
\left\{\begin{array}{@{}l@{}}
   \mathcal{S}^{-1}(\Phi) = M^I \Phi\\
    \Phi_x = U \Phi
  \end{array}\right.
\end{equation*}
and the associated zero curvature condition is
\begin{equation}
 M^I_x = U_{-1} M^I - M^I U.
\label{ID3}
\end{equation}
Comparing \eqref{ID3} with \eqref{D3}, the Darboux inverse matrix $M^I$ does not satisfy equation \eqref{D3}.
If we extend the involution map ${\cI}$ defined by \eqref{involp} to matrices as follows:
\begin{equation}\label{invol}
  \cI(U)=U, \quad\cI(U_{-1})=U_1, \quad \cI(E F)=\cI(E) \cI(F),
\end{equation}
where $E$ and $F$ are matrices such as $M^I$ and $U$, then $\cI(M^I)$ satisfies equation \eqref{D3}.
We will discuss such map for the DNLS equations in Section \ref{inverse}.

We now consider the composition of Darboux transformations. Let $\cS^{(1)}$ and $\cS^{(2)}$ be two Darboux transformations, associated with the Darboux matrices $M$ and $N$, respectively.
We denote the Darboux transformation of $p$ under $\cS^{(1)}$ by $\tilde{p}$, and under $\cS^{(2)}$ by $\bar{p}$.
This notation clarifies the order of successive shifts: applying $\cS^{(2)}$ followed by $\cS^{(1)}$ yields $\cS^{(1)}\cS^{(2)}p = \tilde{\bar{p}}$, while applying $\cS^{(1)}$ followed by $\cS^{(2)}$ yields $\cS^{(2)}\cS^{(1)}p = \bar{\tilde{p}}$. The same notation applies to $q$ and the fundamental solution $\Phi$.
Using this convention, the transformations are defined as:
\begin{equation}\label{TwoDarb}
   \cS^{(1)}: \Phi \mapsto \tilde{\Phi} = M \Phi;\qquad
  \cS^{(2)}: \Phi \mapsto \bar{\Phi} = N \Phi.
\end{equation}
The composition $\cS^{(1)} \cS^{(2)}$ corresponds to the map:
\begin{equation*}
\cS^{(1)} \cS^{(2)} : \Phi \mapsto \widetilde{\bar{\Phi}} = M(\tilde{\bar{p}}, \tilde{\bar{q}}, \bar{p}, \bar{q})  N({\bar{p}}, {\bar{q}}, p, q)  \Phi.
\end{equation*}
Notice that the shift $\cS^{(2)}$ does not act on $M$, but on its arguments.
To better represent the composition of two Darboux transformations, we introduce the operator $\iota_{\cS}$ as
\begin{equation}
 \iota_{\cS} M = M|_{p\mapsto \cS(p), q \mapsto \cS(q)},
 \label{iota}
\end{equation}
and call $\iota_{\cS}$ the ``inner shift".

With this notation, the composition of two Darboux matrices in \eqref{TwoDarb} can be expressed as
\begin{equation}
 \cS^{(1)}  \cS^{(2)} : \Phi \mapsto \tilde{\bar{\Phi}}= \iota_{\cS^{(2)}} (M) N \Phi.
 \label{combina}
\end{equation}
The resulting matrix $\iota_{\cS^{(2)}} (M) N$ is still a Darboux matrix, in the sense that it satisfies the zero curvature condition \eqref{D3} with shift $ \cS^{(1)} \cS^{(2)} $. Indeed, we have
$$ (\iota_{\cS^{(2)}} (M) N)_x = \cS^{(1)} \cS^{(2)} (U) \iota_{\cS^{(2)}} (M) N - \iota_{\cS^{(2)}} (M) N U .$$

\subsection{Darboux transformations for DNLS equations}\label{S4}
Consider a $2 \times 2$ matrix polynomial in the spectral parameter $\lambda$
\begin{equation}
M = \sum_{i=0}^N M^{(i)} \lambda^i,
\label{ansatz0}
\end{equation}
where $N\in \mathbb{N}$ denotes the polynomial degree of $M$, corresponding to the degree of the Darboux transformation.
Since gauge transformations \eqref{gauge} allow rescaling powers of $\lambda$, the degree $N$ is defined under the requirement that
$M^{(0)}\neq 0$.

We assume that $M$ is a Darboux matrix.
Substituting $M$ and $U$ given by \eqref{LaxP1-A} into the zero curvature condition \eqref{D3}, from the coefficient of the highest power in the spectral parameter, i.e., $\lambda^{N+2}$, it follows that matrix $M^{(N)}$, coefficient of $\lambda^N$ in \eqref{ansatz0} (the leading term of $M$), is diagonal.

Knowing that $U$ is invariant under the reduction group in \eqref{reduction1}, we assume that the Darboux matrix \eqref{ansatz0} satisfies
\begin{equation}
M(\lambda) = \pm I M(-\lambda) I,
 \label{reduction2}
\end{equation}
where ``$+$'' corresponds to the original reduction group and ``$-$'' reflects the requirement of nonzero $M^{(0)}$.
By direct computation, we have the following statement:
\begin{prop}\label{prop2}
A polynomial Darboux matrix of degree $N$ \eqref{ansatz0} for the DNLS equations, invariant under the reduction group \eqref{reduction2}, takes the form
\begin{equation}
\begin{split}
M(N) = &\begin{pmatrix}
        M^{(N)}_{11} & 0 \\ 0 &  M^{(N)}_{22}
      \end{pmatrix}\lambda^{N}
      + \begin{pmatrix}
       0 &  M^{(N)}_{11} p - p_1  M^{(N)}_{22} \\ q_1  M^{(N)}_{11} -  M^{(N)}_{22} q & 0
      \end{pmatrix} \lambda^{N-1} \\
      &+ \sum_{j=2}^{N} M_{\sigma(j)}^{(N-j)} \lambda^{N-j}; \qquad
      \sigma(k) = \left\{\begin{array}{@{}l@{}}
        d \qquad k \text{ is even}  \\
        a \qquad k \text{ is odd}
      \end{array}\right.
\end{split}
\label{ansatz3}
\end{equation}
The function $\sigma(k)$ determines whether the coefficient matrix $M^{(k)}$ is diagonal $M^{(k)}_d$ or off-diagonal $M^{(k)}_a$.
\end{prop}

A Darboux transformation of even degree is invariant with respect to \eqref{reduction2} for the plus sign, while if the degree is odd, it corresponds to the minus sign.\par

When $M$ is a polynomial in $\lambda$, we define the rank of such Darboux transformation as the rank of its leading term, which is the matrix coefficient $M^{(N)}$ of the highest power in $\lambda$ \cite{92}.
We thus split \eqref{ansatz3} into three cases:
\begin{itemize}
 \item[{\rm (i)}]  $M^{(N)}_{11}\neq 0$ and $M^{(N)}_{22} = 0$;
\item[{\rm (ii)}]  $M^{(N)}_{11}=0$ and $M^{(N)}_{22}\neq 0$;
\item[{\rm (iii)}]  $M^{(N)}_{11} M^{(N)}_{22} \neq 0$.
\end{itemize}
The first two cases are rank-1 Darboux matrices, denoted by $M_\uparrow(N)$ and $M_\downarrow(N)$, respectively.
As argued for \eqref{DNLS0I} in \cite{80}, cases (i) and (ii) are gauge equivalent, and case (iii) often arises from composing Darboux matrices from the first two cases.

In this paper, we primarily focus on case {\rm (i)} when $N\leq 2$,  corresponding to \eqref{DarBoux1} and \eqref{OriginalMatrix}.
Here, we introduce the auxiliary function $f$ for $M_{11}^{(1)}$ in $M_\uparrow(1)$ and analogously for $M_{11}^{(2)}$ in $M_\uparrow(2)$.
Moreover, we show that for $N\in \{1, 2\}$, the Darboux matrix $M_\downarrow(N)$ relates to the inverse of $M_\uparrow(N)$.

\subsection{Quasideterminants and factorisation}\label{factor}
The reverse process of composition is factorisation. It addresses whether a Darboux transformation of higher polynomial degree
can be expressed as a composition of a certain number of transformations of lower degree.

Recall that in the commutative case, the determinant $\det(M)$ of a Darboux matrix $M$ satisfies
\begin{equation}
  \det(M)^{-1}\det(M)_x = {\rm tr}(U_1) - {\rm tr}(U),
\end{equation}
where its derivative depends on the trace of the associated Lax representation.
The Lax representations of the (commutative) models \eqref{DNLS0I} and \eqref{DNLS0III} are traceless, therefore the determinant of the related $M$ is a constants of motion \cite{6}.

Moreover,  if a Darboux matrix $M$ is the composition of two Darboux matrices $M^{(1)}$ and $M^{(2)}$ in this order, by \eqref{combina} it follows that $M = \iota_{\cS^{(2)}} (M^{(1)}) M^{(2)}$, where $\cS^{(i)}$ is the shift associated with $M^{(i)}$ for $i = 1, 2$.
Consequently, the determinant of $M$ factorises in a product of the determinants of $M^{(i)}$:
$$\det(M)=\iota_{\cS^{(2)}}(\det(M^{(1)}))\det(M^{(2)}).$$
Such observation indicates that the entire transformation can be decomposed into two simpler steps.

For matrices with non-commutative entries, there is no universal definition of determinant although several generalisations that preserve different aspects or properties of the commutative determinant exist. Examples with many applications to integrable system \cite{144,4} include the quasideterminants, introduced and studied in \cite{56}, and the Dieudonn\'e determinant, presented in \cite{141}.

We recall some basic facts on quasideterminants. Consider an invertible $2 \times 2$ matrix $X$ with non-commutative entries:
\begin{equation*}
  X = \begin{pmatrix}
        X_{11} & X_{12} \\ X_{21} & X_{22}
      \end{pmatrix}, \qquad  X^{-1} = \begin{pmatrix}
        \Delta_{11}^{-1}(X) & -\Delta_{11}^{-1}(X) X_{12} X_{22}^{-1}  \\ -\Delta_{22}^{-1}(X) X_{21} X_{11}^{-1}& \Delta_{22}^{-1}(X)
      \end{pmatrix}.
\end{equation*}
The matrix $X^{-1}$ depends on the inverse of two objects $\Delta_{11}(X)$ and $\Delta_{22}(X)$, which are the quasideterminants $\Delta_{ij}(X)$ of $X$ in the positions $(1,1)$ and $(2,2)$, respectively.

For the $2\times 2$-matrix $X$ with non-commutative entries, there are four quasideterminants:
\begin{equation}\label{quasim}
\begin{array}{c}
 \Delta_{11}(X) = X_{11} - X_{12} X_{22}^{-1} X_{21}, \qquad \Delta_{12}(X)= X_{12} - X_{11} X_{21}^{-1} X_{22},\\
 \Delta_{21}(X) = X_{21} - X_{22} X_{12}^{-1} X_{11}, \qquad \Delta_{22}(X)= X_{22} - X_{21} X_{11}^{-1} X_{12} .
 \end{array}
\end{equation}
When all the entries of $X$ are non-zero, the quasideterminants $\Delta_{ij}(X)$, $i, j = 1, 2$ are interrelated.
For instance, all $\Delta_{ij}(X)$ can be expressed in terms of $\Delta_{11}(X)$:
    \begin{align*}
        &\Delta_{12}(X) = - \Delta_{11}(X)X_{21}^{-1}X_{22}, && \Delta_{21}(X) = - X_{22}X_{12}^{-1}\Delta_{11}(X),\\
        & \Delta_{22}(X) =  X_{22}X_{12}^{-1}\Delta_{11}(X)X_{11}^{-1}X_{12}, && \Delta_{22}(X) =  X_{21}X_{11}^{-1}\Delta_{11}(X)X_{21}^{-1}X_{22}.
  \end{align*}
This property allows the inverse of $X$ to be written in two equivalent forms, based on either $\Delta_{11}$ or $\Delta_{22}$:
\begin{equation*}
 X^{-1} = \begin{pmatrix}
        \Delta_{11}^{-1}(X) & -\Delta_{11}^{-1}(X) X_{12} X_{22}^{-1}  \\ -X_{22}^{-1} X_{21}\Delta_{11}^{-1}(X) & X_{22}^{-1} X_{21} \Delta_{11}^{-1}(X) X_{11} X_{21}^{-1}
      \end{pmatrix} = \begin{pmatrix}
        X_{11}^{-1} X_{12}\Delta_{22}^{-1}(X) X_{22} X_{12}^{-1} & -X_{11}^{-1} X_{12} \Delta_{22}^{-1}(X)  \\ -\Delta_{22}^{-1}(X) X_{21} X_{11}^{-1}& \Delta_{22}^{-1}(X)
      \end{pmatrix}
\end{equation*}
In case some of the entries are zero (and the matrix is still invertible), only one inverse form becomes applicable, and can be computed from the definition.\par
\begin{prop}\label{prop3}
The quasidetermiants for $M_\uparrow(1)$ given in \eqref{DarBoux1} and $M_\uparrow(2)$ given in \eqref{OriginalMatrix} are
\begin{eqnarray}\label{quasim1}
 \nexists \Delta_{11}(M_\uparrow(1)), \quad \Delta_{12}(M_\uparrow(1)) = f p, \quad \Delta_{21}(M_\uparrow(1)) = q_1 f, \quad \Delta_{22}(M_\uparrow(1)) = - q_1 f p \lambda^{-1}
\end{eqnarray}
and
\begin{equation}\label{quasim2}
\begin{array}{l}
  \Delta_{11}(M_\uparrow(2)) = a + (\hh - \hh p b^{-1} q_1 \hh)\lambda^2, \\
\Delta_{12}(M_\uparrow(2)) = (\hh p - q_1^{-1} b )\lambda- a \hh^{-1} q_1^{-1} b \lambda^{-1},\\
\Delta_{21}(M_\uparrow(2)) = (q_1 \hh - b p^{-1})\lambda - b p^{-1} \hh^{-1} a \lambda^{-1},\\
\Delta_{22}(M_\uparrow(2)) = q_1 \hh (\hh \lambda^2 + a)^{-1}((q_1^{-1}b - \hh p)\lambda^2 + a \hh^{-1} q_1^{-1}b).
\end{array}
\end{equation}
\end{prop}
\begin{proof}
 These are obtained by directly using formula \eqref{quasim}.
\end{proof}
Unlike the commutative case, the quasideterminants of a product of matrices in the non-commutative framework are not equal to the product of the corresponding quasideterminants. However, for two $n \times n$ matrices $X$ and $Y$ with non-zero entries, the following relation holds \cite{56}:
\begin{equation*}
 \Delta_{ij}^{-1}(X Y) = \sum_{k = 1}^n \Delta_{k j}^{-1}(Y)\Delta_{i k}^{-1}(X).
 \label{Retak}
\end{equation*}
Note that if some entries are zero, the quasideterminants of a product can be expressed in a simplified form \cite{144}.

Given a Darboux transformation of higher degree, we investigate whether it can be obtained by composing two lower-degree transformations. In this paper, we focus on rank-1 Darboux transformations. Inspired by the commutative case, we examine the quasideterminants of the composition of a Darboux matrix $N$ with a linear matrix $M_\uparrow(1)$, considering both left-multiplication and right-multiplication.
\begin{prop}\label{prop4}
Let $\cS_{\uparrow}: \Phi \mapsto  \tilde{\Phi}=M_\uparrow(1) \Phi$ and $\cS_N: \Phi \mapsto \bar{\Phi}=N \Phi$. Then
\begin{eqnarray}
 &&\Delta_{1 i}\left(\iota_{\cS_N} (M_\uparrow(1)) N\right) =\mathcal{S}_N (f p) \Delta_{2 i}(N) , \label{Fac1}\\
 &&\Delta_{i 1}\left(\iota_{\cS_\uparrow} (N) M_\uparrow(1)\right) = \iota_{\mathcal{S}_{\uparrow}} ( \Delta_{i 2} (N)) \mathcal{S}_{\uparrow}(q) f ,\label{Fac2}
\end{eqnarray}
where $\Delta_{ij}$ are quasidetermiants at positions $(i,j)$, $i,j=1,2$, defined by \eqref{quasim}.
\end{prop}
\begin{proof}
Recall that  $M_\uparrow(1)$ is given in \eqref{DarBoux1} and let Darboux matrix $N$ be generic, that is,
    \begin{equation*}
        M_\uparrow(1) = \begin{pmatrix}
            f \lambda & f p \\ q_1 f & 0
        \end{pmatrix}, \qquad N = \begin{pmatrix}
            A & B \\ C & D
        \end{pmatrix}.
    \end{equation*}
For the Darboux transformation $\cS_\uparrow\cS_N$, the associated Darboux matrix is
    \begin{equation*}
        M = \iota_{\cS_N} (M_\uparrow(1)) N = \begin{pmatrix}
            \lambda \bar f A + \bar f \bar p C & \lambda \bar f B + \bar f \bar p D\\ \tilde{\bar q} \bar f A & \tilde{\bar q} \bar f B
        \end{pmatrix}.
    \end{equation*}
According to \eqref{quasim}, the quasideterminants of matrix $M$ are
    \begin{equation*}
        \begin{array}{l} \Delta_{11}(M) = \bar{f}\bar{p}(C - D B^{-1} A) = \cS_N(f p )\Delta_{21}(N),\\
             \Delta_{12}(M) = \bar{f}\bar{p}(D - C A^{-1} B) = \cS_N(f p )\Delta_{22}(N).
        \end{array}
    \end{equation*}
Thus we obtain \eqref{Fac1} as a product of two factors: an entry of $M_\uparrow(1)$ and a quasideterminant of $N$.

Analogously, for the Darboux transformation $ \cS_N\cS_\uparrow$, the associated Darboux matrix is
    \begin{equation*}
        M = \iota_\mathcal{\cS_\uparrow} (N) M_\uparrow(1) = \begin{pmatrix}
            \lambda \iota_\mathcal{\cS_\uparrow} (A) f + \iota_\mathcal{\cS_\uparrow} (B) \tilde{q} f & \iota_\mathcal{\cS_\uparrow} (A) f p \\ \lambda \iota_\mathcal{\cS_\uparrow} (C) f + \iota_\mathcal{\cS_\uparrow} (D) \tilde{q} f & \iota_\mathcal{\cS_\uparrow} (C) f p
        \end{pmatrix}.
    \end{equation*}
    In a similar way as above, the quasideterminants of matrix $M$ are
    \begin{equation*}
        \begin{split}
        & \Delta_{11}(M) = (\iota_\mathcal{\cS_\uparrow} (B) - \iota_\mathcal{\cS_\uparrow} (A) \iota_\mathcal{\cS_\uparrow} (C)^{-1} \iota_\mathcal{\cS_\uparrow} (D)) \tilde{q}f = \iota_\mathcal{\cS_\uparrow} (\Delta_{12}(N)) \tilde{q} f,\\
            & \Delta_{21}(M) =  (\iota_\mathcal{\cS_\uparrow} (D) - \iota_\mathcal{\cS_\uparrow} (C) \iota_\mathcal{\cS_\uparrow} (A)^{-1} \iota_\mathcal{\cS_\uparrow} (B)) \tilde{q}f = \iota_\mathcal{\cS_\uparrow} (\Delta_{22}(N)) \tilde{q} f.
        \end{split}
    \end{equation*}
    As before, these quasideterminants $\Delta_{11}(M)$ and $\Delta_{21}(M)$ have a simple polynomial form and
    correspond to \eqref{Fac2} in the statement.
\end{proof}
\begin{remark}
It is possible to prove an analogous statement involving a rank-1 down linear Darboux matrix:
$$
M_\downarrow(1) = \begin{pmatrix}
      0 & 0 \\ 0 & g
     \end{pmatrix} \lambda - \begin{pmatrix}
     0 &  p_1 g \\
      g q  & 0
\end{pmatrix} .
$$
Considering $\cS_\downarrow: \Phi \mapsto  M_\downarrow(1) \Phi $,  we state that
\begin{eqnarray}
 &&\Delta_{2 i}(\iota_{\cS_N} (M_\downarrow(1)) N) =- \mathcal{S}_N (g q) \Delta_{1 i}(N) , \label{Fac3}\\
 &&\Delta_{i 2}(\iota_{\cS_\downarrow} (N) M_\downarrow(1)) = - \iota_{\mathcal{S}_{\downarrow}} ( \Delta_{i 1} (N)) \mathcal{S}_{\downarrow}(p) g .\label{Fac4}
\end{eqnarray}
\end{remark}
In Section \ref{sec45}, we are going to apply this result to determine whether a quadratic Darboux transformation is a composition of two linear Darboux transformation.

\section{Computation of the Darboux transformations}\label{S6}
In this section, we analyse constant, linear and quadratic Darboux matrices for the DNLS equations, focusing on the $M_\uparrow(N)$
for $N\leq 2$.
After fixing the form of $M$, we proceed applying it to the zero curvature condition \eqref{D3} to derive the governing system.
We then determine the first integrals, reduce the resulting equations, and obtain the associated integrable D$\Delta$Es.

\subsection{Constant Darboux transformations}
We present the case of a polynomial Darboux matrix with only one term in the expansion \eqref{ansatz3}. Due to the scaling seen in \eqref{gauge}, we can eliminate the dependence on $\lambda$, therefore this case is equivalent to a zero degree polynomial (or constant) Darboux transformation:
\begin{equation}\label{MS}
M = \begin{pmatrix}
      a & 0 \\ 0 & b
     \end{pmatrix}, \quad a b\neq 0.
\end{equation}
From the zero curvature condition \eqref{D3} with the Lax representation \eqref{LaxP1},  we obtain the equations
\begin{equation*} 
 \left\{\begin{array}{@{}l@{}}
  a p-p_1 b=0, \qquad  a_{x}=2\mathcal{S}(P_{11}) a-2a P_{11},\\ 
   b q-q_1 a=0,  \qquad  b_{x}=2\mathcal{S}(P_{22}) b-2b P_{22}.
\end{array}\right. 
\end{equation*} 
For the cases \eqref{DNLSI}, \eqref{DNLSII} and \eqref{DNLSIII}, this leads to
\begin{equation*}
 a_x = 0, \qquad b_x = 0.
\end{equation*}
Thus, for these three DNLS equations the constant matrix \eqref{MS} is a Darboux matrix with $a$ and $b$ non-commutative constants.
The resulting B\"acklund transformation corresponds to the scaling symmetry
\begin{equation*}
p_1=a\ {p}\ b^{-1}, \qquad q_1=b\ {q}\ a^{-1}.
\end{equation*}
For \eqref{DNLSII-2} and {\color{black}\eqref{DNLSIII-2}} we obtain
\begin{equation*}
  a_x = 2 (b q p b^{-1} a - a q p), \qquad b_x = 0,
\end{equation*}
while with \eqref{DNLSI-2} and \eqref{DNLSII-3} we have
\begin{equation*}
  a_x = 0, \qquad b_x = 2 (a p q a^{-1} b - b p q).
\end{equation*}
When both $a$ and $b$ are commutative, the Darboux matrix \eqref{MS} has constant evolution for all the DNLS equations and it corresponds with a scaling of the dependent variables $p$ and $q$. The same result for equation \eqref{DNLS0I} is obtained in \cite{92, 80}.
\subsection{Linear Darboux transformations}\label{Lin}
Consider the rank-1 up Darboux matrix $M_\uparrow(1)$ linear in $\lambda$ defined in \eqref{DarBoux1}:
\begin{equation*}
 M_\uparrow(1) = \begin{pmatrix}
      f & 0 \\ 0 & 0
     \end{pmatrix} \lambda + \begin{pmatrix}
     0 &  f   p \\
      q_1   f & 0
\end{pmatrix}.
\end{equation*}
Here we introduce the auxiliary function $f$ in place of the entry $M_{11}^{(1)}$.
In the context of the nonlinear Schr\"odinger equation (NLS), this transformation, called ``elementary'', has been studied in \cite{102, 98}. In \cite{80} the linear Darboux transformations of the equation \eqref{DNLS0I}
are treated as a special case of the quadratic Darboux matrix when $a=b=0$ in \eqref{OriginalMatrix}.

From the zero curvature condition \eqref{D3} with matrix $U$ given by \eqref{LaxP1},
we obtain the system:
\begin{equation}
 \left\{\begin{array}{@{}l@{}}
   p_x = 2\left(P_{11} p - p P_{22}\right) + 2\left( p q p - f^{-1} p_1 q_1 f p\right)\\ 
   q_{1,x} = 2\mathcal{S}\left(P_{22} q - q P_{11} \right) + 2\left(q_1 f p q f^{-1} - q_1 p_1 q_1\right)\\
   f_x= 2\mathcal{S}\left(P_{11}+ p q\right) f - 2 f \left(P_{11} + p q \right)
\end{array}\right.
\label{DT2}
\end{equation}
In the following parts, we analyse this system for all seven DNLS equations one by one.

\subsection*{Equation \eqref{DNLSI}}
We consider equation \eqref{DNLSI} with Lax representation \eqref{LaxA1},
the system \eqref{DT2} becomes
\begin{equation}\label{A1eq}
   \left\{\begin{array}{@{}l@{}}
   f_{x}=2(p_1q_1f -f pq) \\
    p_{x}= 2 (p   q p - f^{ -1}  p_1  q_1  f   p )  \\
   q_{1,x}= 2 (q_1  f   p   q  f^{-1}-q_1  p_1  q_1)
 \end{array}\right.
\end{equation}
This system then yields two constants of motion $ \alpha = f p$ and $\beta = q_1 f$, that is $\alpha_x=\beta_x=0$.
These invariants generate distinct reductions, which we present below for this case.

\begin{itemize}
 \item {\bf Only constant $\bm\alpha$}. We set $f = \alpha p^{-1}$, where $\alpha$ is defined on integers and takes value in the same free associates algebra as the dependent variables $p$ and $q$. Substituting it into \eqref{A1eq}, we get
 \begin{equation}\label{A1dd}
   \left\{\begin{array}{@{}l@{}}
      p_x=   2   (p q p - p \alpha^{-1} p_1 q_1 \alpha) \\
      q_{1,x} = 2 (q_1 \alpha q p \alpha^{-1} -q_1 p_1 q_1)
  \end{array}\right.
 \end{equation}
and its Lax pair is formed by the matrix $U$ given in \eqref{LaxA1} and the Darboux matrix
\begin{equation}\label{DTA1}
   M=\begin{pmatrix}
  \alpha p^{-1} & 0 \\ 0 & 0
 \end{pmatrix}\lambda +  \begin{pmatrix}
    0 &  \alpha \\ q_1 \alpha p^{-1} & 0
    \end{pmatrix}.
\end{equation}
 \item {\bf Only constant $\bm\beta$}. By setting $q_1 f=\beta$, we obtain system \eqref{A1dd} with $\alpha=\beta^{-1}$ up to an involution $\star$ defined by \eqref{transpose}.
 \item {\bf Both constants $\bm\alpha$ and $\bm\beta$}. We employ both constants of motion to determine $q_1 = \beta p \alpha^{-1}$, then, with a gauge transformation
$G ={\rm diag}(1,\ \beta_{-1})$ acting on matrices $U$ and $-M$ in \eqref{DTA1}, we obtain the modified Volterra equation \eqref{mVb} with $u = -p \alpha^{-1}$ and $\mu=-\alpha \beta_{-1}$,
 with Lax pair given by
\begin{equation}\label{MUmVL1}
 M=\begin{pmatrix}
  u^{-1} & 0 \\ 0 & 0
 \end{pmatrix}\lambda +  \begin{pmatrix}
    0 &  \mu \\ -1 & 0
    \end{pmatrix},\qquad  U = \begin{pmatrix}
           1 & 0 \\ 0 & -1
          \end{pmatrix} \lambda^2 + 2 \begin{pmatrix}
          0 & u \mu \\ - u_{-1} & 0
          \end{pmatrix}\lambda.
\end{equation}
\end{itemize}

\subsection*{Equation \eqref{DNLSI-2}}
We consider equation \eqref{DNLSI-2} with Lax representation \eqref{LaxA2}.
Thus the system \eqref{DT2} admits one constant of motion $ \phi = f$ satisfying $\phi_x=0$ and it becomes
\begin{equation*}
    \left\{\begin{array}{@{}l@{}}
            p_x = 2(p^2 q - \phi^{-1}p_1 q_1 \phi p)\\
            q_{1,x} = 2( q_1 \phi p q \phi^{-1} - p_1 q_1^2)
        \end{array}\right. 
\end{equation*}
Introducing a new variable
$$u =- \phi  p q,$$
it follows that
\begin{equation*}
 M=\begin{pmatrix}
 \phi & 0 \\
 0 & 0 \\
\end{pmatrix} \lambda
+
\begin{pmatrix}
 0 & \phi p \\
 -p^{-1}_1 \phi^{-1}_1 u_1 \phi & 0 \\
\end{pmatrix}
\end{equation*}
and the matrix $U$ for \eqref{LaxA2} becomes
\begin{equation*}
U= \begin{pmatrix}
 1 & 0 \\
 0 & -1 \\
\end{pmatrix}\lambda ^2  + 2
\begin{pmatrix}
 0 & p \\
 -p^{-1} \phi^{-1} u & 0 \\
\end{pmatrix} \lambda +2
\begin{pmatrix}
 \phi^{-1} u & 0 \\
 0 & \phi^{-1} u \\
\end{pmatrix}.
\end{equation*}
The above $M$ and $U$ provide a Lax representation for the system
\begin{equation*}
    u_x = 2 (\phi_{1}^{-1} u_1 u -  u u_{-1}\phi_{-1}^{-1}), \qquad p_x = 2 \left( \phi^{-1} \phi_1^{-1} u_1 \phi p- p \phi^{-1} u\right).
\end{equation*}
The evolution of $u$ with $\mu = \phi^{-1}$ is Volterra equation \eqref{Va}, while the evolution of $p$ is linear and it depends on both $u$ and $p$.
The latter can be eliminated by a gauge transformation with $G = {\rm diag}(1,\ p^{-1}\phi^{-1})$, leading to the Lax representation of the Volterra equation \eqref{Va}:
\begin{equation}\label{VLA2}
\begin{split}
 & M = \begin{pmatrix}
            \phi & 0 \\ 0 & 0
        \end{pmatrix} \lambda + \begin{pmatrix}
            0 & 1 \\ -u_1 \phi & 0
        \end{pmatrix},\\
        & U = \begin{pmatrix}
            1 & 0 \\ 0 & -1
        \end{pmatrix}\lambda^2+2\begin{pmatrix}
            0 & \phi^{-1} \\ -u & 0
        \end{pmatrix}\lambda +2\begin{pmatrix}
            \phi^{-1} u & 0 \\ 0 & \phi_1^{-1} u_1
        \end{pmatrix} .
\end{split}
\end{equation}

\subsection*{Equation \eqref{DNLSII}}
We consider equation \eqref{DNLSII} with Lax representation \eqref{LaxB1}.
As for equation \eqref{DNLSI-2}, $\phi=f$ is a constant of motion. Thus the system \eqref{DT2} becomes
\begin{equation}
\label{B1eq}
    \begin{cases}
           & p_x = -2 \phi^{-1}p_1 q_1 \phi p\\
           & q_{1,x} = 2 q_1 \phi p q \phi^{-1}
        \end{cases}
\end{equation}
Introduce the variable
$$u = - \phi p q.$$
The system above then reduces to Volterra equation \eqref{Va} with $\mu = \phi^{-1}$. The associated Lax pair, obtained through a gauge transformation $G = {\rm diag}(1,\ (\phi p )^{-1})$, is identical to \eqref{VLA2}.

System \eqref{B1eq} can also be reduced using its constant of motion, which is given by
\begin{equation}
\label{kappa}
    \kappa =- q_1 \phi \phi_{-1} p_{-1}, \qquad \kappa_x = 0.
\end{equation}
Solving for $q_1$ leads to $q_1 = -\kappa p_{-1}^{-1}\phi_{-1}^{-1}\phi^{-1}$. Substituting it into system \eqref{B1eq},
it follows
 \begin{equation}
 \label{B1dd}
      p_x=  2 \phi^{-1} p_1 \kappa p_{-1}^{-1}\phi_{-1}^{-1} p.
 \end{equation}
When $\kappa=\mathbbm{1}$,  it reduces to the potential Volterra equation \eqref{pV} (see in Appendix \ref{Appb-V}).

Equation \eqref{B1dd} becomes the modified Volterra equation \eqref{mVb} with $\mu = \kappa_1$
under the substitution $u=\kappa p_{-1}^{-1}\phi_{-1}^{-1} p \kappa_1^{-1}$.
The associated Lax representation matches \eqref{MUmVL1}
after a gauge transformation $G = {\rm diag}(p \kappa_1^{-1} u^{-1},\ 1)$.

As shown in Appendix \ref{Appb-V}, equation \eqref{B1dd} can also be transformed into \eqref{Va}.
Introducing the variable $u = - p \kappa_{-1} p_{-2}^{-1} \phi_{-2}^{-1}$, we get \eqref{Va} with $\mu = \phi_{-1}^{-1}$. The associated Lax representation, obtained via a gauge transformation $G = {\rm diag}(\phi_{-1},\ p^{-1})$ and shifting $\phi \mapsto \phi_1$,
is identical with \eqref{VLA2}.

\subsection*{Equation \eqref{DNLSII-2}}
We consider equation \eqref{DNLSII-2} with Lax representation \eqref{LaxB2}.
Thus the system \eqref{DT2} becomes
\begin{equation*}
    \left\{\begin{array}{@{}l@{}}
         f_{x}=2 (p_1 q_1 - q_1 p_1) f -2 f (p q- q p)\\
         p_{x}=2\left(p q -q p -f^{-1} p_1 q_1 f \right)p\\
         q_{1,x}= 2 q_1\left( q_1 p_1 - p_1 q_1+ f p q f^{-1}\right)
    \end{array}\right.
\end{equation*}
The quantity $\kappa$ in \eqref{kappa} is a constant of motion for this system. Using $\kappa$ to determine $q_1=-\kappa p_{-1}^{-1} f_{-1}^{-1} f^{-1}$ and substituting it into the system above, we get
 \begin{equation}
 \label{B2dd}
   \left\{\begin{array}{@{}l@{}}
       f_x = 2 (f p \kappa _{-1} p_{-2}^{-1} f_{-2}^{-1} f_{-1}^{-1}- p_1 \kappa  p_{-1}^{-1} f_{-1}^{-1}- f \kappa _{-1} p_{-2}^{-1} f_{-2}^{-1} f_{-1}^{-1} p+\kappa  p_{-1}^{-1} f_{-1}^{-1} f^{-1} p_1 f)\\
      p_x=  2( f^{-1} p_1 \kappa  p_{-1}^{-1} f_{-1}^{-1}+\kappa _{-1} p_{-2}^{-1} f_{-2}^{-1} f_{-1}^{-1} p- p \kappa _{-1} p_{-2}^{-1} f_{-2}^{-1} f_{-1}^{-1}  )p
  \end{array}\right.
 \end{equation}
Introducing a new variable
$$u=\kappa p_{-1}^{-1} f_{-1}^{-1}p \kappa_1^{-1},$$
system \eqref{B2dd} is transformed into the modified Volterra equation \eqref{mVb} with $\mu = \kappa_1$. The Lax representation is the same as in \eqref{MUmVL1}, obtained under the gauge transformation $G = {\rm diag}(p \kappa_1^{-1} u^{-1},\ 1)$.

\subsection*{Equation \eqref{DNLSII-3}}
We consider equation \eqref{DNLSII-3} with Lax representation \eqref{LaxB3}.
System \eqref{DT2} admits $\phi=f$ as a constant motion, therefore it becomes
\begin{equation*}
    \left\{\begin{array}{@{}l@{}}
        p_{x}= 2 \left( p^2 q-  p q p- \phi^{-1}p_1 q_1 \phi p\right)\\
        q_{1,x}=  2 \left(q_1 p_1 q_1- p_1 q_1^2 + q_1 \phi p q \phi^{-1}\right)
    \end{array}\right.
\end{equation*}
By letting
$$u=-\phi p q,$$
the system reduces to Volterra equation \eqref{Va} with $\mu = \phi^{-1}$. The Lax representation, obtained through a gauge transformation by $G = {\rm diag}(1,\ p^{-1}\phi^{-1})$ is the same as \eqref{VLA2}.

\subsection*{Equation \eqref{DNLSIII}}
We consider equation \eqref{DNLSIII} with Lax representation \eqref{LaxC1},
System \eqref{DT2} admits $\phi=f$ as a constant motion, therefore it becomes
\begin{equation*}
    \left\{\begin{array}{@{}l@{}}
         p_{x}=-2\left(p q  +\phi^{-1} p_1 q_1 \phi \right)p\\
         q_{1,x}=2 q_1\left( p_1 q_1+ \phi p q \phi^{-1}\right)
    \end{array}\right.
\end{equation*}
Introducing a new variable
$$u = - \phi p q,$$
the system reduces to Volterra equation \eqref{Va} with $\mu = \phi^{-1}$. The associated Lax pair, obtained though a gauge transformation
$G = {\rm diag}(1,\ (\phi p )^{-1})$ is identical to \eqref{VLA2}.

The above system can also be reduced using its constant of motion:
\begin{equation} \label{rhol}
    \rho = q_1 \phi p.
\end{equation}
Let $q_1=-\rho p^{-1} \phi^{-1}$, it reduces to
$$
p_x=  2( \phi^{-1} p_1 \rho + p \rho_{-1} p_{-1}^{-1} \phi_{-1}^{-1} p),
$$
which is equivalent to the Volterra equation \eqref{Va} with $\mu = \phi^{-1}$ under the transformation $u=\phi p \rho_{-1} p_{-1}^{-1} \phi_{-1}^{-1}$. The associated Lax representation is the same as \eqref{VLA2} via a gauge transformation $G={\rm diag}(1,\ p^{-1}\phi^{-1})$.

\subsection*{Equation \eqref{DNLSIII-2}}
We consider equation \eqref{DNLSIII-2} with Lax representation \eqref{LaxC2}.
Thus the system \eqref{DT2} becomes
\begin{equation*}
    \left\{\begin{array}{@{}l@{}}
            f_x = 2(f q p - q_1 p_1 f)\\
            p_x = -2(q p+ f^{-1}p_1 q_1 f )p\\
            q_{1,x} = 2q_1(q_1 p_1 + 2 f p q f^{-1})
    \end{array}\right.
\end{equation*}
The quantity $\kappa$ in \eqref{kappa} is a constant of motion for this system. Using $\kappa$ to determine $q_1=-\kappa p_{-1}^{-1} f_{-1}^{-1} f^{-1}$ and substituting it above, we get a system for $f$ and $p$:
 \begin{equation}
 \label{C2dd}
   \left\{\begin{array}{@{}l@{}}
            f_x = 2 \left(\kappa  p_{-1}^{-1} f_{-1}^{-1} f^{-1} p_{1} f-f \kappa _{-1} p_{-2}^{-1} f_{-2}^{-1} f_{-1}^{-1} p\right)\\
            p_{x} = 2 \left(f^{-1} p_{1} \kappa  p_{-1}^{-1} f_{-1}^{-1} +\kappa _{-1} p_{-2}^{-1} f_{-2}^{-1} f_{-1}^{-1} p\right)p
        \end{array}\right.
 \end{equation}
We introduce a new dependent variable
$$u=-\kappa p_{-1}^{-1} f_{-1}^{-1}p \kappa_1^{-1},$$
system \eqref{C2dd} is equivalent to the modified Volterra equation \eqref{mVb} with $\mu = \kappa_1$. The associated Lax representation is the same as \eqref{MUmVL1}, through a gauge transformation by $G = {\rm diag}(p \kappa_1^{-1} u^{-1},\ 1)$.
\subsection*{Summary}
In studying linear Darboux transformations for all seven integrable non-Abelian DNLS equations, the associated D$\Delta$Es \eqref{DT2} reduce to two scalar equations: the Volterra equation \eqref{Va} and the modified Volterra equation \eqref{mVb}. This reduction relies on constants of motion and the introduction of new dependent variables that depend on both $p$ and $q$, e.g. $u=-\phi pq$.
We list the constants of motion below:
 \begin{align*}
 &\mbox{constants used} && \mbox{DNLS equation}&& \mbox{reduced scalar equation}\\
     & \alpha =f p, \, \beta = q_1 f && \eqref{DNLSI} && \qquad \eqref{mVb}\\
   &  \phi = f&&   \eqref{DNLSI-2}, \eqref{DNLSII}, \eqref{DNLSII-3},\eqref{DNLSIII} &&\qquad \eqref{Va}\\
    &  \kappa =  q_1 f f_{-1} p_{-1}&&  {\eqref{DNLSII}, \eqref{DNLSII-2},\eqref{DNLSIII-2}} &&\qquad \eqref{mVb}\\
    &  \rho = q_1 f p && \eqref{DNLSIII} &&\qquad \eqref{Va}
 \end{align*}
The Volterra equations presented here involve non-commutative constants; setting these to $\mathbbm{1}$, the unity of the algebra,  \eqref{mVb} reduces to the modified Volterra ($\mathrm{mVL}^2$) at zero parameter and \eqref{Va} to the Volterra ($\mathrm{VL}^1$) system discussed by Adler in \cite{112}.
Further details on the relationships between these Volterra equations are provided in Appendix \ref{Appb-V}.

In the construction of linear Darboux transformations, one notices the recurrence of analogous or identical structures across the DNLS equations. This observation can be justified by analysing \eqref{DT2} more generally. We illustrate this with the following two cases, which cover all seven DNLS equations.

\begin{enumerate}
    \item[(i)] {\bf Case $\bm {P_{11}=- p q}$}.
Consider the system \eqref{DT2} associated with a up linear Darboux transformation. Applying the change the variables
\begin{equation*}
    u = pq,
\end{equation*}
immediately yields the system
\begin{equation*}
       \left\{\begin{array}{@{}l@{}}
         f_x= 2\mathcal{S}\left(u+ P_{11}\right) f - 2 f \left(u + P_{11}\right)\\
         u_x = 2\left(P_{11} u-u P_{11} \right)+2 \left(u f_{-1} u_{-1}f_{-1}^{-1}-f^{-1} u_1 f u\right)
    \end{array}\right.
\end{equation*}
Notice that the term $P_{22}$ vanishes completely. This change of variables is particularly effective in case $P_{11} = - p q$, as in \eqref{DNLSI-2}, \eqref{DNLSII}, \eqref{DNLSII-3} and \eqref{DNLSIII}, the variable $f$ becomes the constant $\phi$ and the whole system is equivalent to \eqref{Va} under the mapping $u \mapsto - \phi^{-1}u$. The associated Lax representation, corresponding with \eqref{VLA2}, is obtained from $(U,M_\uparrow(1))$ by a gauge transformation $G = {\rm diag}(1,\ (f p)^{-1})$.
\item[(ii)] {\bf Case $\bm {P_{22}=0}$}.
We consider a different change of variables
$$ \kappa = - q_1 f f_{-1}p_{-1}, \qquad u = q_1 f_1^{-1}q_2^{-1}, $$
where $\kappa$ is the same as in \eqref{kappa}.
System \eqref{DT2} expressed in the new variables yields
\begin{equation*}
 \left\{\begin{array}{@{}l@{}}
  \kappa_x = 2\cS\left(P_{22}\right)\kappa - 2 \kappa \cS^{-1}\left(P_{22}\right)\\
   u_x = 2u\left(u_1 \kappa_2 - \kappa_1 u_{-1}\right) u - 2 (u \kappa_1 P_{22} \kappa_1^{-1} - \kappa \cS^{-1}(P_{22}) \kappa^{-1} u)
 \end{array}\right.
\end{equation*}
When $P_{22} = 0$, the variable $\kappa$ becomes a constant of motion, as in \eqref{DNLSI}, \eqref{DNLSII}, \eqref{DNLSII-2} and \eqref{DNLSIII-2}. Similarly, the evolution of $u$ corresponds to \eqref{mVb} with $\mu = \kappa_1$. Applying a gauge transformation $G = {\rm diag}(p \kappa_1^{-1}u_1^{-1},\ 1)$, we get the associated Lax representation \eqref{MUmVL1}.
\end{enumerate}

This scheme of reductions, based on Lax representations with either $P_{11} = -p q$ or $P_{22} = 0$, is not the only possibility. Equivalent situations arise when considering $P_{11} = 0$ or $P_{22} = q p$, yielding similar results. When $P_{11} = 0$, the constant $\theta = f p q f_{-1}$ appears, and the consequent reduction leads to \eqref{mVb}. Conversely, when $P_{22} = q p$, the system admits the constant $\rho = q_1 f p$ (which also appears in the linear Darboux transformation for \eqref{DNLSIII}), resulting in a reduction to the $\cI$ involution of \eqref{Va}.

Many Lax representations in our list fit into more than one scheme. Consider, for instance, \eqref{LaxB1}, which is characterised by both $P_{11} = -p q$ and $P_{22} = 0$. In fact, the system in this case reduces to the potential Volterra equation \eqref{pV}, which can be transformed to either \eqref{Va} or \eqref{mVb} by Miura transformations as demonstrated in Appendix \ref{Appb-V}.

\subsection{Quadratic Darboux transformations}\label{quad}
In this section, we construct quadratic Darboux transformations of the DNLS equations. As we did for the linear transformations in Section \ref{Lin}, we first consider a rank-1 up Darboux matrix $M_\uparrow(2)$ as defined in \eqref{OriginalMatrix}
\begin{equation*}
 M_\uparrow(2)= \begin{pmatrix}
      \hh & 0 \\ 0 & 0 
     \end{pmatrix} \lambda^2 + \begin{pmatrix}
     0 &  \hh  p \\  q_1   \hh & 0
     \end{pmatrix} \lambda + \begin{pmatrix}
      a & 0 \\ 0 & b
     \end{pmatrix}.
     \label{OriginalMatrix0}
\end{equation*}
Note that, together with the dependent variables $p$ and $q$, $M_\uparrow(2)$ includes three auxiliary functions $\hh, a$ and $b$.
This case has been analysed for the commutative model \eqref{DNLS0I} in \cite{92, 80}.

Substituting $M_\uparrow(2)$ in the zero curvature condition \eqref{D3} with the generic representation $U$ in \eqref{LaxP1}, we obtain the following system:
\begin{equation}
 \left\{\begin{array}{@{}l@{}}
   p_x = 2 \hh^{-1} \left(p_1   b - a   p\right) + 2\left( P_{11} p - p P_{22}\right) + 2 (p q  - \hh^{-1}p_1 q_1 \hh)p\\ 
    q_{1,x} = 2 \left(q_1   a - b   q\right)\hh^{-1} +  2\mathcal{S}\left( P_{22} q - q P_{11}\right) - 2 q_1 (p_1 q_1 - \hh p q \hh^{-1})\\ 
    \hh_x=  2\mathcal{S}\left(P_{11} + p  q\right) \hh -2 \hh \left(P_{11} + p q\right)\\
   a_x=  2\mathcal{S}(P_{11}) a - 2a P_{11}\\
   b_x= 2 \mathcal{S}(P_{22}) b - 2 b P_{22}
\end{array}\right.
\label{DT3}
\end{equation}
Define a function $\rho$ by
\begin{equation}
 \rho = q_1 f p - b ,
 \label{tardip}
\end{equation}
it follows that
\begin{equation*}
 \rho_x = 2\mathcal{S}(P_{22}-q p)\rho - 2 \rho (P_{22}- q p) .
 \label{tardi}
\end{equation*}
Thus $\rho$ is a constant of motion for all those models associated with a Lax representation where $P_{22} = q p$ (in this case only \eqref{LaxC1} satisfies this condition). However, equation \eqref{tardi} is trivially satisfied by $\rho = 0$, therefore the assumption $b = q_1 \hh p$ is a valid reduction for all DNLS equations.
In Section \ref{sec45}, we interpret $\rho = 0$ as a necessary condition for $M_{\uparrow}(2)$ to decompose into two linear up Darboux matrices.

System \eqref{DT3} is evolutionary for the dependent variables $p,q$ and the three auxiliary functions $f, a$ and $b$.
To get systems with reduced numbers of dependent variables,
we carry out the following three reductions:
\begin{equation}\label{3cases}
{\rm (i)\ only}\ a =0; \qquad {\rm (ii)\ only}\ b = 0;\qquad {\rm (iii)\ only}\ \rho = 0.
\end{equation}
Clearly we can not set either $\hh =0$ nor $a=b=0$, as either reduction would cause $M_\uparrow(2)$ to reduce to the linear matrix $M_\uparrow(1)$.
When $\rho = a = 0$, this results in a trivial Darboux transformation. Thus we also exclude it.

We will analyse the effects of the above reductions on equation \eqref{DNLSI}. For all other DNLS equations, analogous results obtained through similar reasoning will be summarised without detailed discussion.

\subsubsection*{Equation \eqref{DNLSI}}
We consider equation \eqref{DNLSI} with Lax representation \eqref{LaxA1}.
It follows from system  \eqref{DT3} that the auxiliary functions $a$ and $b$ are constant, denoted respectively by $\alpha$ and $\beta$. Consequently, \eqref{DT3} reduces to the following system where $f$ is the only auxiliary function:
\begin{equation}
    \left\{\begin{array}{@{}l@{}}
   \hh_{x}=2(p_1q_1 \hh -\hh pq) \\
    p_{x}=2 \hh^{ -1}(p_1\beta - \alpha p ) + 2 (p   q  - \hh^{ -1}  p_1  q_1  \hh  )p  \\
   q_{1,x}= 2(q_1  \alpha   -  \beta q ) \hh^{-1}+2 q_1(  \hh   p   q  \hh^{-1}-  p_1  q_1)
 \end{array}\right.
 \label{bohhhhhh}
\end{equation}
We now carry out the three reductions listed in \eqref{3cases}.
\begin{enumerate}
 \item[(i)] {\bf Only $\bm {\alpha=0}$}.
 We introduce two new dependent variables
 $$u = q_1   \hh   p, \qquad v = - q   \hh^{-1}   q_1^{-1},$$
  and use them to express $\hh$ and $q$. This leads to
    \begin{equation}
    \left\{\begin{array}{@{}l@{}}
   u_{x}=2v_1 u_1 (u-\beta )-2(u-\beta ) v u\\
    v_{x}=2 (v (u-\beta ) -(u_{-1}-\beta _{-1}) v_{-1} )v  \\
   p_{x}=2 p (u^{-1} v_1 u_1 (u-\beta )-v u)
 \end{array}\right.
 \label{SIa0eq}
\end{equation}
with a Lax representation given by
\begin{equation}
    \begin{split}
        & M=
\begin{pmatrix}
 -p_1 u_1^{-1} v_1^{-1} u p^{-1} & 0 \\
 0 & 0 \\
\end{pmatrix}\lambda ^2
+
\begin{pmatrix}
 0 & -p_1 u_1^{-1} v_1^{-1} u \\
 u p^{-1} & 0 \\
\end{pmatrix}\lambda+
\begin{pmatrix}
 0 & 0 \\
 0 & \beta  \\
\end{pmatrix},
\nonumber\\
  & U =
\begin{pmatrix}
 1 & 0 \\
 0 & -1 \\
\end{pmatrix}\lambda ^2 +2
\begin{pmatrix}
 0 & p \\
 -v u p^{-1} & 0 \\
\end{pmatrix}\lambda.
    \end{split}
\end{equation}
Notice that for system \eqref{SIa0eq} the dynamics of $u$ and $v$ depend on $u$ and $v$ only, while the dynamics of $p$ is linear in $p$.
The system of $u$ and $v$  becomes \eqref{N1} under $u \mapsto u + \beta$ and $\mu = \beta$.
As shown in Appendix \ref{Appb-Q1}, $\eqref{N1}$ is related to Ablowitz-Ladik lattice \eqref{AL} via a Miura transformation.

The presence of $p$ in the resulting Lax representation can be removed by a gauge transformation $G={\rm diag}( p u^{-1},\ 1)$.
Recalling $u \mapsto u+\beta$, we obtain a Lax representation of \eqref{N1} as follows:
\begin{equation}
\begin{split}
  & M=\begin{pmatrix}
  v_1^{-1} & 0 \\ 0 & 0
 \end{pmatrix}\lambda^2 +  \begin{pmatrix}
    0 &  v_1^{-1} (u + \beta) \\ -1 & 0
    \end{pmatrix} \lambda- \begin{pmatrix}
    0 & 0 \\ 0 & \beta
    \end{pmatrix},\\
    & U = \begin{pmatrix}
           1 & 0 \\ 0 & -1
          \end{pmatrix}\lambda^2 + 2 \begin{pmatrix}
          0 & u + \beta \\ -v & 0
          \end{pmatrix} \lambda+ 2 \begin{pmatrix}
          \beta v & 0 \\ 0 & 0
          \end{pmatrix}.
    \end{split}\label{0018}
\end{equation}
\item[(ii)] {\bf Only $\bm {\beta=0}$}.
Define two new variables
$$u = - p_1 q_1, \qquad v = \hh^{-1}.$$
As before, we use $u$ and $v$ to express $q$ and $\hh$. The variable $p$ in the Lax representation is eliminated via a gauge transformation by the matrix $G={\rm diag}(1,\ p^{-1} v)$. This leads to
\begin{equation}
\begin{split}
  & M=\begin{pmatrix}
  v^{-1} & 0 \\ 0 & 0
 \end{pmatrix}\lambda^2 + \begin{pmatrix}
    0 & 1 \\ -v_1^{-1} u v^{-1} & 0
    \end{pmatrix} \lambda +  \begin{pmatrix}
    \alpha & 0 \\ 0 & 0
    \end{pmatrix},\\
    & U = \begin{pmatrix}
           1 & 0 \\ 0 & -1
          \end{pmatrix}\lambda^2 + 2  \begin{pmatrix}
          0 & v \\ -v^{-1} u_{-1} & 0
          \end{pmatrix} \lambda- 2 \begin{pmatrix}
          0 & 0 \\ 0 & \alpha v
          \end{pmatrix}.
\end{split}
 \label{DT2NLSIc}
\end{equation}
The resulting integrable system corresponds with \eqref{N2} when $\mu = \al$.
Taking the Abelian version of \eqref{N2} with the transformation $u \mapsto -u$ yields the relativistic Toda lattice \cite{174,175}. Since \eqref{N2} is linked with \eqref{N1} as shown in Appendix \ref{Appb-Q1}, it is consequently linked with the Ablowitz-Ladik lattice \eqref{AL}.
\item[(iii)] {\bf Only $\bm {\rho=0}$}. Since $b=\beta$ for \eqref{DNLSI}, it follows from $\rho=0$ that $\hh = q_1^{-1}\beta  p^{-1}$.
Thus system \eqref{bohhhhhh} reduces to a system for $p$ and $q$:
\begin{equation}
\left\{\begin{array}{@{}l@{}}
      p_{x}=  2p (q -  \beta^{-1} q_1 \alpha ) p\\
      q_{1,x}= q_1 (\alpha p \beta^{-1} - p_1 ) q_1
  \end{array}\right.
  \label{DT2NLSIaa}
\end{equation}
with a Lax pair given by
\begin{eqnarray*}
  && M= \begin{pmatrix}
  q_1^{-1} \beta p^{-1} & 0 \\ 0 & 0
 \end{pmatrix}\lambda^2 + \begin{pmatrix}
    0 & q_1^{-1} \beta \\ \beta p^{-1} & 0
    \end{pmatrix}\lambda +  \begin{pmatrix}
     \alpha & 0 \\ 0 & \beta
    \end{pmatrix},\\
    && U = \begin{pmatrix}
           1 & 0 \\ 0 & -1
          \end{pmatrix}\lambda^2 + 2  \begin{pmatrix}
          0 & p \\ q & 0
          \end{pmatrix}\lambda.
\end{eqnarray*}
Let
$$ u = - \beta^{-1} q_1 \alpha p, \qquad v = - \beta_{-1}^{-1} q p.$$
System \eqref{DT2NLSIaa} becomes the two-component Volterra lattice \eqref{2Vb} with $\mu=\beta_{-1}$ (we refer Appendix \ref{Appb-V} for the different versions of non-Abelian Volterra equations). Employing a gauge transformation $G={\rm diag}(p,\  \beta_{-1})$, we obtain a Lax representation for \eqref{2Vb}:
\begin{equation}
    \begin{split}
  & M=\begin{pmatrix}
    v_1^{-1} & 0 \\ 0 & 0
 \end{pmatrix}\lambda^2 +  \begin{pmatrix}
    0 &  v_1^{-1}\beta_{-1} \\ -1 & 0
    \end{pmatrix} \lambda-  \begin{pmatrix}
    v_1^{-1} u & 0 \\ 0 & \beta_{-1}
    \end{pmatrix},\\
    & U = \begin{pmatrix}
           1 & 0 \\ 0 & -1
          \end{pmatrix}\lambda^2 + 2  \begin{pmatrix}
          0 & \beta_{-1} \\ - v & 0
          \end{pmatrix}\lambda + 2 \begin{pmatrix}
          \beta_{-1}v- u & 0 \\ 0 & 0
          \end{pmatrix}.
    \end{split}\label{0019}
\end{equation}
\end{enumerate}
\subsubsection*{Equation \eqref{DNLSI-2}}
From the Lax representation \eqref{LaxA2} and the system \eqref{DT3}, it follows that $\hh_x=0$. Let $\hh=\phi$, and system \eqref{DT3} becomes
\begin{equation}
\label{A2qdd}
    \left\{\begin{array}{@{}l@{}}
         a_{x}=2 (a p q- p_1 q_1 a)\\
         b_{x}=2 (b p q- p_1 q_1 b)\\
         p_{x}=2 \phi^{-1}\left( p_1 b- a p\right)+2\left(p^2 q-\phi^{-1} p_1 q_1 \phi p\right)\\
         q_{1,x}=2 \left(q_1 a- b q \right)\phi^{-1}+2 \left( q_1 \phi p q \phi^{-1}- p_1 q_1^2\right)
    \end{array}\right.
\end{equation}
Again we carry out the reductions in the three cases listed in \eqref{3cases}.
\begin{enumerate}
 \item[(i)] {\bf Only $\bm {a=0}$}.
We introduce the transformation
$$
  u = \phi^{-1}\left(p_1 q_1 \phi - v_1\right) , \qquad v = p b_{-1} p_{-1}^{-1}
$$
and use them to replace $b$ and $q$ in \eqref{A2qdd}. The reduced system of $u$ and $v$ is transformed into \eqref{rT} when $\mu = \phi^{-1}_{-1}$.
The variable $p$ in the Lax representation can be eliminated with the gauge matrix $G = {\rm diag}(1,\ p^{-1})$, which leads to a Lax representation for \eqref{rT}:
\begin{equation}
\begin{split}
     & M=
\begin{pmatrix}
 \phi & 0 \\
 0 & 0 \\
\end{pmatrix}\lambda ^2 
+
\begin{pmatrix}
 0 & \phi \\
 \phi u+v_1 & 0 \\
\end{pmatrix}\lambda 
+
\begin{pmatrix}
 0 & 0 \\
 0 & v_1 \\
\end{pmatrix},
\\
        & U=\begin{pmatrix}
 1 & 0 \\
 0 & -1 \\
\end{pmatrix}
\lambda ^2 + 2 
\begin{pmatrix}
 0 & 1 \\
 \left(v+\phi_{-1} u_{-1}\right) \phi^{-1}_{-1} & 0 \\
\end{pmatrix}\lambda 
-2 
\begin{pmatrix}
 \left(v+\phi_{-1} u_{-1}\right) \phi^{-1}_{-1} & 0 \\
 0 & u \\
\end{pmatrix}.
    \end{split}
    \label{A2a0}
\end{equation}
 \item[(ii)] {\bf Only $\bm {b=0}$}.
Consider the new variables
$$u = p q, \qquad v =a_{-1}.$$
We rewrite system \eqref{A2qdd} in terms of $u, v$ and $p$. The resulting system of $u$ and $v$ is
the relativistic Toda \eqref{rT} with $\mu = \phi^{-1}_{-1}$.
The variable $p$ in the Lax representation can be removed by a gauge transformation with $G = {\rm diag}(1,\ p^{-1})$, which leads to a Lax representation for \eqref{rT}.
\begin{equation}
\begin{split}
    & M= 
\begin{pmatrix}
 \phi & 0 \\
 0 & 0 \\
\end{pmatrix}
\lambda ^2+ 
\begin{pmatrix}
 0 & \phi \\
 u_1 \phi & 0 \\
\end{pmatrix}
\lambda+
\begin{pmatrix}
 v_1 & 0 \\
 0 & 0 \\
\end{pmatrix},
\\
        & U=
\begin{pmatrix}
 1 & 0 \\
 0 & -1 \\
\end{pmatrix} \lambda ^2
+  2 
\begin{pmatrix}
 0 & 1 \\
 u & 0 \\
\end{pmatrix}
\lambda-2
\begin{pmatrix}
 u & 0 \\
 0 & \phi^{-1} \left(u_1 \phi+v_1\right) \\
\end{pmatrix}.
    \end{split}
    \label{001}
\end{equation}
 \item[(iii)] {\bf Only $\bm {\rho=0}$}.
Substituting $b=q_1 \phi p$ into system \eqref{A2qdd}, and introducing the change of variables
$$u = - p q, \qquad v = - a_{-1},$$
we obtain \eqref{2Va} with $\mu = \phi^{-1}_{-1}$.
The variable $p$ can be removed by a gauge transformation with $G = {\rm diag}(1,\ p^{-1})$. This leads to a Lax representation
for \eqref{2Va}:
\begin{equation}
\begin{split}
     & M=
\begin{pmatrix}
 \phi & 0 \\
 0 & 0 \\
\end{pmatrix}\lambda ^2
+
\begin{pmatrix}
 0 & \phi \\
 -u_1 \phi & 0 \\
\end{pmatrix}
\lambda -
\begin{pmatrix}
 v_1 & 0 \\
 0 & u_1 \phi \\
\end{pmatrix},\\
    &  U= \begin{pmatrix}
 1 & 0 \\
 0 & -1 \\
\end{pmatrix}\lambda ^2  
 +2
\begin{pmatrix}
 0 & 1 \\
 -u & 0 \\
\end{pmatrix}
\lambda +  2 
\begin{pmatrix}
 u & 0 \\
 0 & \phi^{-1} v_1 \\
\end{pmatrix}.
\end{split}
\label{009}
\end{equation}
\end{enumerate}
\subsubsection*{Equation \eqref{DNLSII}}
It follows from the Lax representation \eqref{LaxB1} and the system \eqref{DT3} that $\hh_x=0$ and $b_x = 0$. We denote them as $\phi=f$ and $\beta=b$.
The system \eqref{DT3} becomes
    \begin{equation}
     \left\{\begin{array}{@{}l@{}}
   a _{x}=2(a p q-p_1 q_1 a)\\
   p_{x}= 2 \phi^{-1} (p_1 \beta - a p ) -2 \phi^{-1} p_1 q_1 \phi p\\
   q_{1,x}= 2 (q_1 a - \beta q) \phi^{-1} + 2 q_1 \phi p q \phi^{-1}
 \end{array}\right.
 \label{DT33}
\end{equation}
\begin{enumerate}
 \item[(i)] {\bf Only $\bm {a=0}$}.
Considering the change of variables
$$u = \phi p, \qquad v = q_1,$$
we obtain the Ablowitz-Ladik system \eqref{AL} with $\mu = \phi^{-1}$ and $\nu = \beta$. The associated Lax pair is
\begin{equation}
    \begin{split}
        & M=
\begin{pmatrix}
 \phi & 0 \\
 0 & 0 \\
\end{pmatrix}\lambda ^2+ 
\begin{pmatrix}
 0 &  u \\
 v \phi & 0 \\
\end{pmatrix}  \lambda +
\begin{pmatrix}
 0 & 0 \\
 0 & \beta \\
\end{pmatrix},\\
    &  U=
\begin{pmatrix}
 1 & 0 \\
 0 & -1 \\
\end{pmatrix}\lambda ^2  + 2  
\begin{pmatrix}
 0 & \phi^{-1} u \\
 v_{-1} & 0 \\
\end{pmatrix} \lambda- 2 
\begin{pmatrix}
\phi^{-1} u v_{-1} & 0 \\
 0 & 0\\
\end{pmatrix}.
    \end{split}\label{004bis}
\end{equation}
 \item[(ii)] {\bf Only $\bm {\beta=0}$}.
 Considering the change of variables
 $$u =p q, \qquad v = a_{-1},$$
 the reduced system is equivalent to \eqref{rT} with $\mu = \phi^{-1}_{-1}$.
The variable $p$ can be removed by a gauge transformation with $G={\rm diag}(1,\ p^{-1})$, the resulting Lax representation is identical to \eqref{001}.
 \item[(iii)] {\bf Only $\bm {\rho=0}$}.
 This reduction suggest to set $q_1=\beta p^{-1} \phi^{-1}$. Considering the change of variables
 $$u = -p \beta_{-1} p_{-1}^{-1} \phi^{-1}_{-1}, \qquad v = - a_{-1},$$
the resulting system is equivalent to \eqref{2Va} with $\mu = \phi^{-1}_{-1}$.
The variable $p$ can be removed via the gauge transformation $G = \mathrm{diag}(1,\ p^{-1})$, resulting in a Lax representation identical to \eqref{009}.
\end{enumerate}
\subsubsection*{Equation \eqref{DNLSII-2}}
It follows from the Lax representation \eqref{LaxB2} and system \eqref{DT3} that $b_x=0$. We denote it as $\beta=b$.
Now system \eqref{DT3} becomes
\begin{equation*}
    \left\{\begin{array}{@{}l@{}}
         a_{x}=2 (a q p- q_1 p_1 a)\\
         \hh_{x}=2 (p_1 q_1 - q_1 p_1) \hh -2 \hh (p q - q p)\\
         p_{x}=2 \hh^{-1}  (p_1 \beta-a p)+2 (p q  -q p -\hh^{-1} p_1 q_1 \hh)p\\
         q_{1,x}=2 (q_1 a -\beta q )\hh^{-1}+2q_1(q_1 p_1- p_1 q_1 + \hh p q \hh^{-1})
    \end{array}\right.
\end{equation*}
\begin{enumerate}
 \item[(i)] {\bf Only $\bm {a=0}$}.
 Considering the change of variables
 $$ u = q_1 \hh p- \beta, \qquad v = - q \hh^{-1} q_1^{-1},$$
the resulting system of $u$ and $v$ is equivalent to \eqref{N1} with $\mu = \beta$.
The variable $p$ in the Lax representation can be removed by a gauge transformation $G={\rm diag}( p(u + \beta)^{-1},\ 1)$, leading to
a Lax representation identical to \eqref{0018}.
 \item[(ii)] {\bf Only $\bm {\beta=0}$}.
This case admits a new constant of motion
\begin{equation}
    \theta = (q_1 \hh a^{-1}\hh p)^{-1}, \label{theta}
\end{equation}
implying $a=\hh p \theta q_1 \hh$. Then we consider the change of variables
$$u = q p, \qquad v = -q_1 \hh p.$$
It leads to a system of $u$ and $v$ as follows:
\begin{equation}
        \left\{\begin{array}{@{}l@{}}
         u_x = 2(u \theta v - v_{-1} \theta_{-1} u) + 2 (v_{-1} u_{-1}v_{-1}^{-1} u - u v^{-1}u_1 v) \\
         v_x = 2 ( v u - u_1 v)
    \end{array}\right.\label{007}
\end{equation}
which is equivalent to the involution $\cI$ (defined in \eqref{involp}) of \eqref{N2} with $\mu = \theta$.
The variable $p$ can be removed by a gauge transformation $G={\rm diag}( p v^{-1},\ 1)$, resulting in the associated Lax representation
\begin{equation}
 \begin{split}
    & M=\begin{pmatrix}
        v_1 u_{1}^{-1}  & 0 \\ 0 & 0
    \end{pmatrix}\lambda^2+\begin{pmatrix}
        0 & v_1 u_{1}^{-1}  v \\  1 & 0
    \end{pmatrix}\lambda-\begin{pmatrix}
        v_1 u_{1}^{-1}  v \theta & 0 \\ 0 & 0
    \end{pmatrix},\\
    & U=\begin{pmatrix}
        1 & 0 \\ 0 & -1
    \end{pmatrix}\lambda^2+2\begin{pmatrix}
        0 & v \\ u v^{-1}  & 0
    \end{pmatrix}\lambda - 2\begin{pmatrix}
        v \theta & 0 \\ 0 & 0
    \end{pmatrix}.
   \end{split} \label{LaxIN2}
\end{equation}
\item[(iii)] {\bf Only $\bm {\rho=0}$}. This implies that $q_1 \hh p=\beta$, that is $\hh= q_1^{-1} \beta p^{-1}$.
Consider the change of variables
$$ u = - \beta^{-1} q_1 a p, \qquad v = - \beta_{-1}^{-1} q p. $$
Then the system of $u$ and $v$ is equivalent to \eqref{2Vb} with $\mu = \beta_{-1}$. The variable $p$ can be removed by a gauge transformation
$G ={\rm diag}(p,\ \beta_{-1})$, resulting in the Lax representation \eqref{0019}.
\end{enumerate}
\subsubsection*{Equation \eqref{DNLSII-3}}
Following from the Lax representation \eqref{LaxB3} and system \eqref{DT3}, we deduce that $\hh_x=0$. We denote it as $\hh = \phi$.
Now system \eqref{DT3} is rewritten as
\begin{equation*}
    \left\{\begin{array}{@{}l@{}}
       a_{x}=2 (a p q- p_1 q_1 a)\\
        b_{x}=2 (q_1 p_1- p_1 q_1 )b -2 b(q p-p q)\\
        p_{x}= 2 \phi^{-1}(p_1 b-a p)+2 (  p^2 q-  p q p- \phi^{-1}p_1 q_1 \phi p)\\
        q_{1,x}=2 (q_1 a-b q)\phi^{-1} +2 ( q_1 p_1 q_1 - p_1 q_1^2+ q_1 \phi p q \phi^{-1} )
    \end{array}\right.
\end{equation*}
\begin{enumerate}
 \item[(i)] {\bf Only $\bm {a=0}$}.
 Considering the change of variables 
$$u_1 =(\phi p  - q^{-1}_1 b)q, \qquad v = p q,$$
the resulting system for $u$ and $v$ becomes
\begin{equation}
        \left\{\begin{array}{@{}l@{}}
          u_x = 2(u v_{-1} -v u  ) + 2u  ( u_{-1}\phi_{-2}^{-1} -   \phi_{-1}^{-1}u)\\
         v_x = 2(v u \phi_{-1}^{-1}-\phi^{-1} v_{1} u_1)
    \end{array}\right.\label{0011}
\end{equation}
which is equivalent to the involution $\cI$ of \eqref{N3} with $\mu = \phi_{-1}^{-1}$.
The variable $p$ can be eliminated by a gauge transformation $G ={\rm diag}(1,\ p^{-1} v)$, resulting a Lax representation of \eqref{0011}:
\begin{eqnarray*}
    && M=\begin{pmatrix}
        \phi & 0 \\ 0 & 0
    \end{pmatrix}\lambda^2+\begin{pmatrix}
        0 & \phi v \\  \phi & 0
    \end{pmatrix}\lambda+\begin{pmatrix}
        0 & 0 \\ 0 &  \phi  v-u_1
    \end{pmatrix},\\
    && U=\begin{pmatrix}
        1 & 0 \\ 0 & -1
    \end{pmatrix}\lambda^2+2\begin{pmatrix}
        0 & v \\ 1 & 0
    \end{pmatrix}\lambda-2\begin{pmatrix}
        v & 0 \\ 0 &   u \phi_{-1}^{-1}
    \end{pmatrix}.
\end{eqnarray*}
 \item[(ii)] {\bf Only $\bm {b=0}$}.
Considering the change of variables
$$u =p q, \qquad v = a_{-1},$$
the reduced system for $u$ and $v$ becomes \eqref{rT} with $\mu = \phi^{-1}_{-1}$.
The variable $p$ can be eliminated by a gauge transformation $G={\rm diag}(1,\ p^{-1})$, the resulting Lax representation is identical to \eqref{001}.
 \item[(iii)] {\bf Only $\bm {\rho=0}$}. This implies that $q_1 \phi p=b$. Setting $q_1=b p^{-1}\phi^{-1}$ and
considering the change of variables
$$ u = -p q, \qquad v = - a_{-1}, $$
the reduced system for $u$ and $v$ becomes \eqref{2Va} with $\mu = \phi^{-1}_{-1}$. The variable $p$ can be removed by a gauge transformation $G={\rm diag}(1,\ p^{-1})$ and it leads to a Lax representation identical to \eqref{009}.
\end{enumerate}

\subsubsection*{Equation \eqref{DNLSIII}}
Following from the Lax representation \eqref{LaxC1} and system \eqref{DT3}, we deduce that $\hh_x=0$. We denote it as $\hh = \phi$. Now system \eqref{DT3} becomes
\begin{equation}
    \left\{\begin{array}{@{}l@{}}
         a_{x}=2 (a p q-p_1 q_1 a)\\
         b_{x}=2 (q_1 p_1 b- b q p)\\
         p_{x}=2 \phi^{-1} (p_1 b-a p)-2 ( p q + \phi^{-1} p_1 q_1 \phi)p\\
         q_{1,x}=2 (q_1 a- b q )\phi^{-1}+2 q_1 ( p_1 q_1 + \phi p q \phi^{-1})
    \end{array}\right.
    \label{c1dd}
\end{equation}
For this equation, $\rho$ given in \eqref{tardip} is a constant of motion. We define $b = q_1 \phi p - \rho$. Thus the system above is reduced to
\begin{equation}
        \left\{\begin{array}{@{}l@{}}
         a_{x}=2 (a p q- p_1 q_1 a)\\
         p_{x}=-2 \phi^{-1}\left(a p + p_1 \rho + \phi p q p\right)\\
         q_{1,x}=2 \left( q_1 a + \rho q + q_1 p_1 q_1 \phi\right)\phi^{-1}
    \end{array}\right.
    \label{c1dd2}
\end{equation}
\begin{enumerate}
 \item[(i)] {\bf Only $\bm {a=0}$}.
The reduced system from \eqref{c1dd2} becomes equivalent to \eqref{MRT} with $u=p, v=q$,
$\mu = \phi^{-1}$ and $\nu = - \rho$. It follows that a Lax representation in terms of $u$ and $v$ is
\begin{equation}
\begin{split}
    & M=\begin{pmatrix}
        \phi & 0 \\ 0 & 0
    \end{pmatrix}\lambda^2+\begin{pmatrix}
        0 & \phi u  \\  v_1 \phi & 0
    \end{pmatrix}\lambda+\begin{pmatrix}
        0 & 0 \\ 0 & v_1 \phi u  - \rho
    \end{pmatrix},\\
    & U=\begin{pmatrix}
        1 & 0 \\ 0 & -1
    \end{pmatrix}\lambda^2+2\begin{pmatrix}
        0 & u \\ v & 0
    \end{pmatrix}\lambda+2\begin{pmatrix}
        - u v & 0 \\ 0 & v u
    \end{pmatrix}.\\
\end{split}
\label{0012}
\end{equation}
 \item[(ii)] {\bf Only $\bm {b=0}$}.
 Consider the change of variables
 $$ u = p q, \qquad v = a_{-1}. $$
 The reduced system from \eqref{c1dd} becomes \eqref{rT} with $\mu = \phi^{-1}_{-1}$.
 The variable $p$ can be eliminated by a gauge transformation $G={\rm diag}(1,\ p^{-1})$, resulting in a Lax representation identical to \eqref{001}.
 \item[(iii)] {\bf Only $\bm {\rho=0}$}.
 Consider the change of variables
 $$ u = -p q, \qquad v = -a_{-1}. $$
 The reduced system from \eqref{c1dd2} becomes \eqref{2Va} with $\mu = \phi^{-1}_{-1}$. The variable $p$ can be eliminated by a gauge transformation $G={\rm diag}(1,\ p^{-1})$, leading to a Lax representation identical to \eqref{009}.
\end{enumerate}
\subsubsection*{Equation \eqref{DNLSIII-2}}
It follows from the Lax representation \eqref{LaxC2} and system \eqref{DT3} that $b_x=0$. We denote it by $b=\beta$.
Now system \eqref{DT3} becomes
\begin{equation*}
    \left\{\begin{array}{@{}l@{}}
         \hh_{x}=2 \left(\hh q p- q_1 p_1 \hh\right)\\
         a_{x}=2 a( p q+  q p)-2 \left(p_1 q_1 +q_1 p_1 \right)a\\
         p_{x}=2 \hh^{-1}\left( p_1 \beta - a p\right) -2(q p+\hh^{-1} p_1 q_1 \hh )p\\
         q_{1,x}=2 \left(q_1 a- \beta q \right) \hh^{-1}+2q_1(q_1 p_1 +  \hh p q \hh^{-1})
    \end{array}\right.
\end{equation*}
\begin{enumerate}
 \item[(i)] {\bf Only $\bm {a=0}$}.
 Consider the change of variables
 $$ u =  q_1 f p - \beta, \qquad v = q p (u+\beta)^{-1}. $$
 The reduced system of $u$ and $v$ becomes \eqref{N1} with $\mu = \beta$. The variable $p$ can be eliminated by a gauge transformation $G ={\rm diag}(p (u+\beta)^{-1},\  1)$, and the resulting Lax representation is identical to \eqref{0018}.
 \item[(ii)] {\bf Only $\bm {\beta=0}$}.
In this case, the reduction admits the constant of motion $\theta$ defined in \eqref{theta}, which we use to determine $a$
by setting $a=f p \theta q_1 f$. Considering the change of variables
$$ u = qp, \qquad v = -q_1 f p,$$
the reduced system is equivalent to the involution $\cI$ in \eqref{007} of \eqref{N2} with $\mu=\theta$.  The variable $p$ can be eliminated by a gauge transformation $G ={\rm diag}(p v^{-1},\ 1)$, and the resulting Lax representation is identical to \eqref{LaxIN2}.

 \item[(iii)] {\bf Only $\bm {\rho=0}$}. This implies that $q_1 \hh p=\beta$, that is, $\hh= q_1^{-1} \beta p^{-1}$.
 Considering the change of variables
 $$u =  - q_1 a p \beta^{-1}, \qquad v = - q p \beta^{-1}, $$
 the reduced system becomes \eqref{2Vb} with $\mu = \beta$.
 The variable $p$ can be eliminated by a gauge transformation $G ={\rm diag}(p \beta^{-1},\ 1)$, and it leads to a Lax representation identical to \eqref{0019}.
\end{enumerate}
\subsection*{Summary}
The construction of quadratic Darboux matrices for DNLS equations is based on system \eqref{DT3}, which involves the dependent variables $p$ and $q$ along with three auxiliary functions.
For each of the seven integrable non-Abelian DNLS systems, we implement the three reductions specified in \eqref{3cases}. The resulting models are summarised in the following table:
    \begin{align*}
          && a=0 && b=0 && \rho = 0\\
         \eqref{DNLSI} && \eqref{N1} && \eqref{N2} && \eqref{2Vb} \\
         \eqref{DNLSI-2} && \eqref{rT} && \eqref{rT} && \eqref{2Va}  \\
         \eqref{DNLSII} && \eqref{AL} &&  \eqref{rT} && \eqref{2Va}\\
         \eqref{DNLSII-2} && \eqref{N1} && \cI\eqref{N2} && \eqref{2Vb}\\
         \eqref{DNLSII-3} && \cI\eqref{N3} && \eqref{rT} && \eqref{2Va}\\
         \eqref{DNLSIII} && \eqref{MRT} && \eqref{rT} && \eqref{2Va}\\
         \eqref{DNLSIII-2} && \eqref{N1} && \cI\eqref{N2} && \eqref{2Vb}
    \end{align*}
The symbol $\cI$ denotes the involution defined in \eqref{involp}. To the best of our knowledge, systems \eqref{N1} and \eqref{N2} and \eqref{N3} are new. In Appendix \ref{AppD}, we
present Miura transformations mapping \eqref{AL} to \eqref{N1} and \eqref{N2}, and \eqref{MRT} to \eqref{N3}.

As in the linear case in Section \ref{Lin}, many DNLS equations exhibit analogous, sometimes identical, structures. This similarity arises because their reductions are based on Lax matrices $U$ sharing common characteristics. In particular, the reduction procedures for $b = 0$ and $\rho = 0$ are notably similar. We demonstrate this by introducing a parameter $\varepsilon$, taking values $1$ or $0$, and setting $b = \varepsilon q_1 f p$. The case $\varepsilon = 0$ corresponds to the reduction $b = 0$, while $\varepsilon = 1$ yields $\rho = 0$. To clarify the relationship between these reductions, we define the variables
    $$ u = p_1 q_1, \qquad w = f^{-1}. $$
The system \eqref{DT3} associated with $M_\uparrow(2)$ then becomes:
    \begin{equation}
    \left\{\begin{array}{@{}l@{}}
         a_x = 2\cS(P_{11})a -2 a P_{11}\\
         w_x =  2 (P_{11} + u_{-1})w-2w\cS(P_{11} + u_{-1}) \\
         u_x =2 \left(\cS(P_{11}) u-u \cS(P_{11})+u a w-w_1 a_1 u\right) +2(\varepsilon -1) \left(w_1 u_1 w_1^{-1} u-u w^{-1} u_{-1} w\right)
    \end{array}\right.
    \label{500}
\end{equation}
It is clear that when $P_{11} = 0$, system \eqref{500} admits $a$ as a constant of motion, and when $P_{11} = -pq$, it admits $w$ as a constant of motion. In both cases, the system reduces from three equations to two.
\begin{enumerate}
    \item[(i)] {\bf Case $\bm {P_{11}=0}$}. In this case, we have $a_x=0$ and denote it by $\alpha=a$. System \eqref{500} is reduced to
    \begin{equation*}
    \left\{\begin{array}{@{}l@{}}
         w_x =  2\left(  u_{-1}w-wu\right) \\
         u_x =2 \left(u \alpha w-w_1 \alpha_1 u\right)+2(\varepsilon -1) \left(w_1 u_1 w_1^{-1} u-u w^{-1} u_{-1} w\right)
    \end{array}\right.
\end{equation*}
In the reduction $\varepsilon = 0$ (corresponding to $b = 0$), the system above is equivalent to \eqref{N2} under the transformations $u \mapsto -u$, $v \mapsto w$, and $\mu = \alpha$.
Similarly, when $\varepsilon = 1$ (so $\rho = 0$), the resulting system reduces to \eqref{2Vb} under $u \mapsto -u$, $v \mapsto -w$, and $\mu = \alpha$.
Note that \eqref{DNLSI} is the only system with $P_{11} = 0$.

\item[(ii)] {\bf Case $\bm {P_{11}=-pq}$}. In this case, we have $w_x=0$, that is, $f_x=0$. We denote it by $f=w^{-1}=\phi$. System \eqref{500} is then reduced to
\begin{equation*}
    \left\{\begin{array}{@{}l@{}}
         a_x = 2(a u_{-1}-u a)\\
         u_x =2 \left(u a \phi^{-1}-\phi^{-1}_1 a_1 u\right)+2(\varepsilon -1) \left(\phi^{-1}_1 u_1 \phi_1 u-u \phi u_{-1} \phi^{-1}\right)
    \end{array}\right.
\end{equation*}
As above, the reduction $\varepsilon=0$ yields \eqref{rT} under the transformations $a = v$ and $\mu = \phi^{-1}$, while $\varepsilon=1$ yields \eqref{2Va} under the transformations $u \mapsto -u$, $v = -a$, and $\mu = \phi^{-1}$. Systems \eqref{DNLSI-2}, \eqref{DNLSII}, \eqref{DNLSII-3} and \eqref{DNLSIII} all belong to this case.
\end{enumerate}
\subsection{Rank-1 down Darboux matrices}\label{inverse}
In this section, we consider the rank-1 down Darboux transformations generated by the matrices $M_\downarrow(N)$ when $N\in \{1, 2\}$.
We are going to show that (in relevant cases) they are associated with the inverses of analogous $M_\uparrow(N)$.

First, we consider the rank-1 down linear Darboux matrix
\begin{equation}
M_\downarrow(1) = \begin{pmatrix}
      0 & 0 \\ 0 & g
     \end{pmatrix} \lambda - \begin{pmatrix}
     0 &  p_1 g \\
      g q  & 0
\end{pmatrix}.
\label{down}
\end{equation}
Given the respective up case $M_\uparrow(1)$ defined by \eqref{DarBoux1}, its Darboux inverse matrix (defined in \eqref{ID1}) is
\begin{equation*}
M_\uparrow^I(1) = -\begin{pmatrix}
                   0 & 0 \\ 0 & (q f_{-1} p_{-1})^{-1}
                  \end{pmatrix} \lambda+
                  \begin{pmatrix}
                   0 & (q f_{-1})^{-1} \\ (f_{-1} p_{-1})^{-1} & 0
                  \end{pmatrix}.
\end{equation*}
Clearly, the set of equations associated with $M_\uparrow^I(1)$ is still \eqref{DT2}.
Matrix $M_\uparrow^I(1)$ shares the same form as $M_\downarrow(1)$ in \eqref{down}, but with different shifts. We apply the map $\cI$ (see \eqref{involp} and \eqref{invol}) to $M_\uparrow^I(1)$:
\begin{equation*}
\cI M_\uparrow^I(1) = -\begin{pmatrix}
                   0 & 0 \\ 0 & (q f_{1} p_{1})^{-1}
                  \end{pmatrix} \lambda+
                  \begin{pmatrix}
                   0 & (q f_{1})^{-1} \\ (f_{1} p_{1})^{-1} & 0
                  \end{pmatrix}.
\end{equation*}
Introducing a new auxiliary function $ g = -(q f_1 p_1)^{-1}$, we obtain from $\cI M_\uparrow^I(1)$ the same $M_\downarrow(1)$ matrix in \eqref{down}. Thus we have the following statement:
\begin{prop}\label{prop5}
The matrix $M_\downarrow(1)$ given by \eqref{down} with $ g = -(q f_1 p_1)^{-1}$ is a Darboux matrix for a DNLS equation if and only if $M_\uparrow(1)$ defined by \eqref{DarBoux1} is its Darboux matrix.
\end{prop}
Given the above result, the Darboux transformation with matrix $M_\downarrow(1)$ can be obtained from $M_\uparrow(1)$ via the involution map $\mathcal{I}$ and a change of auxiliary function. The associated systems of equations are therefore equivalent under the same operations.

We now address the rank-1 down quadratic Darboux transformation of the form:
\begin{equation} 
M_\downarrow(2) = \begin{pmatrix} 
          0 & 0 \\ 0 & g           
    \end{pmatrix} \lambda^2 - \begin{pmatrix} 
      0 & p_1 g \\ g q & 0  
    \end{pmatrix}\lambda + \begin{pmatrix} 
     c & 0 \\ 0 & d 
    \end{pmatrix}.
\label{M2down} 
\end{equation}
Following our approach for the linear case, we need to examine whether the Darboux inverse $M^I_\uparrow(2)$ of matrix \eqref{OriginalMatrix} is of the same form of $M_\downarrow(2)$.
Direct computation shows that $M^I_\uparrow(2)$ it not a (Laurent) polynomial in $\lambda$, so it can not be compared to $M_\downarrow(2)$ in \eqref{M2down}.
However, when studying $M_\uparrow(2)$, we focus on the three reductions specified in \eqref{3cases}. In these reduced cases, $M^I_\uparrow(2)$ is indeed polynomial as follows:
\begin{enumerate}
 \item[{\rm (i)}] {\bf Case $\bm {a=0}$}. The inverse matrix becomes
$$ M^I_\uparrow(2) =
\begin{pmatrix}
0 & 0 \\
0 & p_{-1}^{-1}  {\Delta} _{-1}^{-1}  \hh_{-1}  p_{-1}  b_{-1}^{-1} \\
\end{pmatrix}
-
\begin{pmatrix}
0 & {\Delta} _{-1}^{-1}  \hh_{-1}  p_{-1}  b_{-1}^{-1} \\
b_{-1}^{-1}  q  \hh_{-1}  {\Delta}_{-1}^{-1} & 0 \\
\end{pmatrix}\lambda^{-1}
+
\begin{pmatrix}
{\Delta}_{-1}^{-1} & 0 \\
0 & 0 \\
\end{pmatrix}\lambda ^{-2},
$$
where ${\Delta} = \hh - \hh p b^{-1} q_1 \hh $. Notice that $\Delta$ here originates from the quasideterminant $\Delta_{11}(M_\uparrow(2))\lambda^{-2}$ when $a=0$. This matrix is an equivalent version of $M_\downarrow(2)$ as follows:
$$ M_\downarrow(2) = \lambda^2 \cI M^I_\uparrow(2), \qquad \text{when } d = 0, $$
and the auxiliary functions $g = (b_1 - q \hh_1 p_1)^{-1}$ and $c = p_1 g b_1 p_1^{-1} \hh_1^{-1}$.
\item[{\rm (ii)}] {\bf Case $\bm {b=0}$}. The inverse matrix becomes
$$ M^I_\uparrow(2)= -
\begin{pmatrix}
0 & 0 \\
0 & (q \hh_{-1} p_{-1})^{-1} \\
\end{pmatrix}
+
\begin{pmatrix}
0 & (q \hh_{-1})^{-1}  \\
( \hh_{-1} p_{-1})^{-1} & 0 \\
\end{pmatrix}\lambda^{-1}
-
\begin{pmatrix}
0 & 0 \\
0 & (q \hh_{-1} a_{-1}^{-1} \hh_{-1}  p_{-1})^{-1} \\
\end{pmatrix}\lambda^{-2}.
$$
When we apply the map $\cI$, we identify this result with the equivalent reduction of $M_\downarrow(2)$ as it follows:
$$ M_\downarrow(2) = \lambda^2 \cI M^I_\uparrow(2), \qquad \text{when }  c = 0  ,$$
and the auxiliary functions $g = -(q f_1 p_1)^{-1}$ and $d = {-} g q a_1 p_1 g$.
\item[{\rm (iii)}] {\bf Case $\bm {\rho=0}$}. In this case, we have $b=q_1 f p$. The inverse matrix becomes
$$ M^I_\uparrow(2)=
\begin{pmatrix}
0 & 0 \\
0 & (q  a _{-1} p_{-1})^{-1} \\
\end{pmatrix} \lambda ^2
-
\begin{pmatrix}
0 & (q a_{-1})^{-1}  \\
(a_{-1} p_{-1})^{-1} & 0 \\
\end{pmatrix} \lambda
+
\begin{pmatrix}
a_{-1}^{-1} & 0 \\
0 & (q \hh_{-1} p_{-1})^{-1} \\
\end{pmatrix},
$$
which can be identified with an equivalent reduction of $M_\downarrow(2)$ as follows:
$$ M_\downarrow(2) = \cI M^I_\uparrow(2), \qquad \text{when }  c = a_1^{-1},$$
and $g =  p_{1}^{-1}  a _{1}^{-1}  q^{-1}$ and $d = p_{1}^{-1}  \hh_{1}^{-1}  q^{-1}$.
The condition $c = a_1^{-1}$ is equivalent as assuming $c=p_1 g q$, which is analogous to $\rho =0$  given in \eqref{tardip}. Thus both matrices belong to the same reduction.

Note that in this case, $M_\uparrow^I(2)$ is polynomial in $\lambda$, similar to the linear Darboux matrix $M_\uparrow^I(1)$. In Section \ref{sec45}, we show that
the reduction $\rho = 0$ is a necessary condition for $M_\uparrow(2)$ to factorise as a composition of two $M_\uparrow(1)$ matrices.
\end{enumerate}

For all three reduction cases carried out, we have demonstrated a connection between $M_\downarrow(2)$ and $M_\uparrow(2)$. In an analogous manner, the associated systems of D$\Delta$Es are shown to be equivalent under a redefinition of the auxiliary functions.

\subsection{Factorisation of quadratic Darboux transformations}\label{sec45}
In this section, we address whether a Darboux transformation of higher degree can be factorised into a composition of lower-degree transformations. Similar results have been obtained in specific cases: for the NLS equation \cite{102, 98}, an entire class of Darboux transformations can be expressed as compositions of elementary Darboux matrices.
We show that for the reduction $\rho=0$ in Section \ref{quad}, the Darboux matrix $M_\uparrow(2)$ decomposes into a composition of two rank-1  linear up Darboux matrices using the quasideterminant properties from Section \ref{factor}.

Consider two rank-1 up linear Darboux matrices in \eqref{DarBoux1}, namely
\begin{equation*}
    M_\uparrow(1) = \begin{pmatrix}
        f \lambda& f p \\ \cS_M(q) f & 0
       \end{pmatrix},\qquad N_\uparrow(1) = \begin{pmatrix}
        g \lambda& g p \\ \cS_N(q) g & 0
       \end{pmatrix},
\end{equation*}
corresponding to the Darboux transformations $\cS_M$ and $\cS_N$, respectively,
and a rank-1 up quadratic Darboux matrix $M_\uparrow(2)$ given in \eqref{OriginalMatrix} of the form
\begin{equation*}
    M_\uparrow(2)= \begin{pmatrix}
      h & 0 \\ 0 & 0
     \end{pmatrix} \lambda^2 + \begin{pmatrix}
     0 &  h   p \\  \cS(q)   h & 0
     \end{pmatrix} \lambda + \begin{pmatrix}
      a & 0 \\ 0 & b
     \end{pmatrix},
\end{equation*}
corresponding to the Darboux transformation $\cS$.

Proposition \ref{prop4} provides a necessary conditions for the decomposition of $M_\uparrow(2)$ into a composition of two Darboux matrices, expressed in terms of quasideterminants. We then examine the three reductions of $M_\uparrow(2)$ discussed in Section \ref{quad} by computing their quasideterminants using Proposition \ref{prop3}.
\begin{itemize}
\item \textbf{Case $\bm{a = 0}$}
    \begin{align*}
        & \Delta_{11}(M_\uparrow(2)) =  (h -  h p b^{-1} \cS(q) h )\lambda^2;\\
        & \Delta_{12}(M_\uparrow(2)) = (h p - \cS(q)^{-1} b )\lambda;\\
        & \Delta_{21}(M_\uparrow(2)) = (\cS(q) h - b p^{-1})\lambda;\\
        & \Delta_{22}(M_\uparrow(2)) =b - \cS(q) h p.
    \end{align*}
 Note that the further reduction $\rho = 0$ reduces all the quasideterminants to zero.
\item \textbf{Case $\bm{b = 0}$}
    \begin{align*}
        & \nexists\Delta_{11}(M_\uparrow(2)); \\
        & \Delta_{12}(M_\uparrow(2)) = h p \lambda;\\
        & \Delta_{21}(M_\uparrow(2)) = \cS(q) h \lambda;\\
        & \Delta_{22}(M_\uparrow(2)) = -\cS(q) h (h \lambda^2 + a)^{-1} h p\lambda^2.
    \end{align*}
\item \textbf{Case $\bm{\rho = 0\ (b=\cS(q) h p)}$}
    \begin{align*}
        & \Delta_{11}(M_\uparrow(2)) = a;\\
        & \Delta_{12}(M_\uparrow(2)) = - a p \lambda^{-1};\\
        & \Delta_{21}(M_\uparrow(2)) =  - \cS(q) a \lambda^{-1};\\
        & \Delta_{22}(M_\uparrow(2)) = \cS(q)  h  (h \lambda^2 + a)^{-1} a p .
    \end{align*}
\end{itemize}
According to \eqref{Fac1} and \eqref{Fac2}, for $M_\uparrow(2)$ to be factorisable into $M_\uparrow(1)$ and $N_\uparrow(1)$, its quasideterminants $\Delta_{11}$, $\Delta_{12}$, and $\Delta_{21}$ must be monomial.  Using Proposition \ref{prop3},
conditions \eqref{Fac1} and \eqref{Fac2} become
    \begin{equation*}
        \left\{\begin{array}{@{}l@{}}
         a + (h - h p b^{-1} \cS(q) h)\lambda^2 = \cS_N( f  p  q) g\\
         (h p - \cS(q)^{-1} b )\lambda- a h^{-1} \cS(q)^{-1} b \lambda^{-1} = -\cS_N(f  p  q ) g p \lambda^{-1}
        \end{array}\right.
    \end{equation*}
which leads to a solution
\begin{equation*}
b = \cS(q) h p, \qquad         a  = \cS_N( f  p q) g.
\end{equation*}

The necessary condition $b = \cS(q) h p$ identifies the reduction $\rho =0$.
The condition $a  = \cS_N( f  p q) g$ is the requirement to identify the auxiliary functions $f$ and $g$ such that $h=\cS_N(f) g$.
For example, for system \eqref{DNLSI}, when we consider its linear Darboux transformations, we know that $fp=\mu$, $\cS_M(q) f=\nu$, $gp=\gamma$ and $\cS_N(q) g=\kappa$ are constants.
Thus we have
$$f=\mu p^{-1}=\cS_M(q)^{-1} \nu, \qquad g=\gamma p^{-1}={\cS_N(q)}^{-1} \kappa ,$$
which is consistent with $a=\cS_N(\mu) \kappa$ and $b=\cS_N(\nu) \gamma$ being constant in this case.

\section{Discussions and further research}
In this paper, we investigate linear and quadratic Darboux transformations for all seven S-integrable DNLS systems defined over a free associative algebra. We derive a list of non-Abelian integrable D$\Delta$Es containing non-Abelian constants. These systems \eqref{AL}--\eqref{2Vb} are presented in the introduction.
Such constants, which exhibit trivial dynamics, are functions on the lattice $\mathbb{Z}$ taking values in the same associative algebra. This nature arises naturally because
they originate from constants of motion for the systems, involving both dependent variables and auxiliary functions, derived from the zero-curvature condition for the Darboux matrices.
If we instead assume these constants are Abelian, specifically, of the form $\gamma \mathbbm{1}$ where $\gamma \in \mathbb{C}$ and $\mathbbm{1}$ is the unit element of the algebra,  they can all be removed through scaling transformations. Under this reduction, we recover the non-Abelian systems studied in \cite{26}.

Quasideterminants play an important role in this work. The integrable D$\Delta$Es presented are obtained through reductions of systems involving auxiliary functions. The higher the degree of Darboux matrices, the more reductions are possible. Our choices are inspired by the quasidetermiants of quadratic Darboux matrices, specifically requiring them to be monomial in terms of spectral parameter $\lambda$. In this way, we obtained all known integrable D$\Delta$Es associated with the NLS equation. This link deserves further study.
Moreover, as we demonstrate in Section \ref{sec45}, quasidetermiants also provide necessary conditions for decomposing a Darboux matrix.

Lax representations are not unique for integrable systems. We choose the matrix $P$ in \eqref{LaxP1-A} with entries either $P_{11}=-p q$ or $P_{22}=0$ from the classification result stated in  Theorem \ref{thm1} in Section \ref{S3}. Different choices result in different Darboux matrices.
The construction of Darboux matrices for the DNLS equation for all Lax representations is presented in full details in \cite{173}. The resulting D$\Delta$Es include those listed in this paper under Miura transformations. In Appendices \ref{Appb-Q1} and \ref{Appb-Q2}, we present Miura transformations related to the integrable D$\Delta$Es we derived.

The Lax representations of all DNLS systems are related by gauge transformations \cite{74}, defined via a diagonal matrix $G = \mathrm{diag}(A, B)$. Here, $A$ and $B$ are non-local (integral) functions satisfying \eqref{gaugeaffect1} and \eqref{gaugeaffect2}.
This gauge structure extends to Darboux transformations as discussed in \eqref{gauge}.
Computationally, this offers significant advantages: determining Darboux transformations for one DNLS equation suffices, since transformations for others follow via gauge equivalence.
However, the gauge-transformed Darboux matrix \eqref{gauge} may lose locality in its entries. It is known that non-local $A$ and $B$ exist which preserve locality in the transformed Lax representation \cite{74}. The corresponding non-locality problem for Darboux matrices remains unaddressed.

The D$\Delta$Es presented in this paper are generalisation of known non-Abelian integrable equations with non-Abelian constants.
Methods for studying the algebraic and geometric properties such as recursion operators and Hamiltonian structures for non-Abelian integrable equations can also be extended to these systems. This extension will be the subject of future research.
Symmetry reductions of these systems result in the Painlev\'e type discrete equations containing non-Abelian constants.
As an example, we consider a symmetry reduction of the Volterra equation \eqref{Va}.
Using its Lax representation, we find the next symmetry of \eqref{Va}
\begin{equation*}
   \begin{split}
        & 2 \mu_1 \left(u_1 u \mu+u_1 \mu_1 u_1+ \mu_2 u_2 u_1\right) u  -2 u \left(u_{-1} u_{-2} \mu_{-2}+u_{-1} \mu_{-1} u_{-1}+\mu u u_{-1}\right) \mu_{-1}.
   \end{split}
\end{equation*}
Following the procedure in \cite{112} and using this symmetry, we derive a generalisation of a non-Abelian discrete Painlev\'e equation ${\rm dP_1}$:
\begin{equation}
    \begin{split}
        & u_n u_{n-1} \mu_{n-1}+u_n \mu_n u_n+\mu_{n+1} u_{n+1} u_n + 4 x u_n + \left( n+\nu + (-1)^n \sigma  \right)\mu_n^{-1} = 0, \quad \nu, \sigma\in \mathbb{C}.
    \end{split}\label{dP1}
\end{equation}
where $\mu$ is a non-commutative constant. If $\mu$ is commutative and we rescale $u_n \mapsto \mu_n^{-1} u_n$, equation \eqref{dP1} becomes
\begin{equation*}
      u_{n+1} u_n  +u_n^2  +u_n u_{n-1}+ 4 x u_n + n-\nu +  (-1)^n \sigma = 0,
     \label{dP12}
\end{equation*}
which first appeared in \cite{166,167} and was recently derived by Adler in \cite{112,168} as a symmetry reduction of the non-Abelian Volterra equation
$
u_{n,x}=u_{n+1} u_{n}-u_n u_{n-1}.
$

The non-Abelian equations can also be viewed as quantised versions of classical systems. Recently, Mikhailov proposed an algebraic quantisation approach based on the notion of quantisation ideals in \cite{27}. One starts with a dynamical system defined on a free associative algebra, where quantisation is understood as a reduction to a system on a quotient algebra over a two-sided ideal invariant under the dynamics.
Crucially, this approach does not rely on a Poisson structure. It has revealed new and surprising quantisation ideals,
as explored in \cite{95} for the case of the Volterra chain.
We aim to investigate whether this framework can also be applied to non-Abelian integrable D$\Delta$Es containing non-Abelian constants.

\section*{Acknowledgements}
EP gratefully acknowledges the support of the EPSRC grant EP/V520093/1.
JPW gratefully acknowledges the support of Ningbo University Research Start-up Fund No. 029-432514993.
Both authors would like to thank A.V. Mikhailov and A. Hone for useful discussions.

This article is partially based upon work from COST Action CaLISTA
CA21109 supported by COST (European Cooperation in Science and Technology). www.cost.eu.

\appendix

\section{Transformations related integrable D\texorpdfstring{$\Delta$}{TEXT}Es}\label{AppD}
In this appendix, we prove that systems \eqref{N1} and \eqref{N2} are related to \eqref{AL} via non-invertible Miura transformations, while system \eqref{N3} is related to \eqref{MRT} and \eqref{K}.
We include graphs summarising the relations among the discussed systems: an arrow represents a Miura transformation from one set of variables to another, and a double-headed arrow indicates that the transformation is invertible.
\subsection{Equations related to the Volterra lattice}\label{Appb-V}
The Volterra-related equations appear recurrently in this paper: in the basic form \eqref{Va}, the modified version \eqref{mVb}, and the two-component systems \eqref{2Va} and \eqref{2Vb}. As noted by Adler \cite{112}, there are two non-Abelian versions of the commutative Volterra equation, which he denoted by $VL^1$ and $VL^2$:
    \begin{align*}
        & u_x = u_1 u - u u_{-1}, \label{Vol1} \tag{${\rm VL^1}$}\\
        & u_x = u_1^\star u - u u_{-1}^\star, \label{Vol2}\tag{${\rm VL^2}$}
    \end{align*}
where $\star$ is the involution defined in \eqref{transpose}.
Adler elaborated a scheme that connects \eqref{Vol1} and \eqref{Vol2} with the potential Volterra equation, two versions of the modified Volterra equation and the corresponding two-component Volterra equations.
In this appendix, we extend this result for the Volterra equations including non-commutative constants.\par
The pivotal equation in this scheme is the potential modified Volterra equation:
\begin{equation*}
    v_x = \mu_{-1}v_1 v_{-1}^{-1}\mu_{-2} v, \tag{${\rm pmV}$}\label{pV}
\end{equation*}
where $\mu$ is a non-commutative constant.
It relates to two sequences of Volterra-type equations, labeled $a$ and $b$. Each sequence comprises:
the Volterra equations \eqref{Va} \cite{158,159},
the modified Volterra equations \eqref{mVa} and \eqref{mVb} \cite{170},
and the two-component Volterra equations \eqref{2Va} and \eqref{2Vb}. They are defined as follows:
    \begin{align*}
        &  w_x=\mu w_1 \mu^{-1}w^{2}- w^{ 2} \mu^{-1}_{-1} w_{-1}\mu _{-1};\tag{${\rm mV_a}$}\label{mVa}\\
        &  u_x = \mu_1 u_1 u - u u_{-1} \mu_{-1};\tag{${\rm V_a}$}\\
        &  r_x = \theta_1 s_1 r - r s \theta,\quad s_x = r s - s r_{-1}, \quad \theta_n=\mu_{2n-1} \mu_{2n};\tag{${\rm 2V_a}$}\\
        &  w_x =  w ( w_1 \eta_1  - \eta w_{-1} ) w, \quad  \eta=\mu_1^{-1} \mu_{-1};  \tag{${\rm mV_b}$}\\
        &  r_x = s_1 \theta_1 r - r \theta s, \qquad s_x = s r - r_{-1} s, \quad \theta_n=\eta_{2n}. \tag{${\rm 2V_b}$}
    \end{align*}
Equation \eqref{Va} generalises Adler's \eqref{Vol1}. Generalising \eqref{Vol2}, however, requires extra conditions on non-commutative constants and lies beyond the scope of this paper.

In the following part we present the changes of variables (Miura transformations) among the given models.
From equation \eqref{pV}, we obtain \eqref{mVa} via the substitution:
\begin{equation*}
    \eqref{pV} \rightarrow \eqref{mVa}: \qquad w = v_1 v^{-1} \mu_{-1}.
\end{equation*}
The modified Volterra equation \eqref{mVa} can be converted into \eqref{Va} using:
\begin{equation*}
    \eqref{mVa} \rightarrow \eqref{Va}: \qquad u = w_1 \mu^{-1} w.
\end{equation*}
Transformations between \eqref{Va} and \eqref{mVa} when $\mu$ is commutative are the well studied Miura maps \cite{158, 161, 26}.
From both \eqref{Va} and \eqref{mVa}, one may deduce \eqref{2Va}. The transformation from \eqref{mVa} is the composition: $\eqref{mVa} \rightarrow \eqref{Va}\rightarrow \eqref{2Va}$:
    \begin{align*}
        & \eqref{mVa} \rightarrow \eqref{2Va}:  && r_n = \mu_{2n+1} w_{2n+2}\mu_{2n+1}^{-1} w_{2n+1}, && s_n = w_{2n+1} \mu_{2n}^{-1}w_{2n} \mu_{2n-1}^{-1}, && \theta_n = \mu_{2n-1}\mu_{2n};\\
        & \eqref{Va} \rightarrow \eqref{2Va}:  && r_n = \mu_{2n+1}u_{2n+1}, && s_n = u_{2n} \mu_{2n-1}^{-1}, && \theta_n = \mu_{2n-1}\mu_{2n}.
    \end{align*}
Note that system  \eqref{2Va} is invariant under the change of variables given by
    $$r' = \theta_1 s_1, \qquad s' = r \theta^{-1}.$$
This property is motivated considering the same transformation on $V_a$: the change of variables becomes equivalent to a shift and a rescaling of the dependent variable $u_{2n}' = \mu_{2n+1} u_{2n+1}\mu_{2n}^{-1}$.

We consider the sequence $b$. The Miura transformation from the potential Volterra \eqref{pV} to the modified Volterra equation \eqref{mVb} is given by
\begin{equation*}
    \eqref{pV} \rightarrow \eqref{mVb}: \qquad w = \mu^{-1}v^{-1}\mu_{-1}v_1 \mu_1, \qquad \eta= \mu_1^{-1}\mu_{-1}.
\end{equation*}
From \eqref{mVb} we obtain the two-components system \eqref{2Vb} introducing
\begin{equation*}
    \eqref{mVb} \rightarrow \eqref{2Vb}: \qquad r_n = w_{2n+1} \eta_{2n+1} w_{2n}, \qquad s_n = w_{2n-1}w_{2n}, \qquad \theta_n = \eta_{2n}.
\end{equation*}
The relationships between these Volterra-related equations are captured in the following graph:
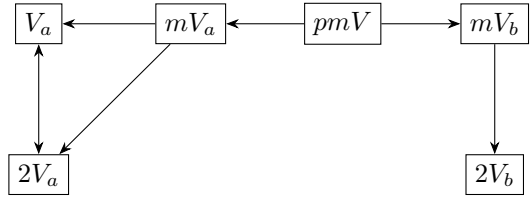
\begin{equation*}
    \begin{tikzpicture}[main/.style = {draw, rectangle}, none/.style={}, none1/.style={draw, shape = circle, fill = black, minimum size = 0.05cm, inner sep=0pt}]
    \node[main] (0) {$pmV$};
    \node[none] (01)[left of=0] {};
    \node[main] (1)[left of=01] {$mV_a$};
    \node[none] (02)[left of=1] {};
    \node[main] (2)[left of=02] {$V_a$};
    \node[none] (03)[below of=2] {};
    \node[main] (3)[below of=03] {$2V_a$};
    \node[none] (04)[right of=0] {};
    \node[main] (4)[right of=04] {$mV_b$};
    \node[none] (05)[right of=4] {};
    \node[none] (06)[below of=4] {};
    \node[main] (6)[below of=06] {$2V_b$};
    \draw[-{Stealth}] (0) -- node[midway, left]{} (1);
    \draw[-{Stealth}] (1) -- node[midway, left]{} (2);
    \draw[{Stealth}-{Stealth}] (2) -- node[midway, left]{} (3);
    \draw[-{Stealth}] (1) -- node[midway, left]{} (3);
    \draw[-{Stealth}] (0) -- node[midway, left]{} (4);
    \draw[-{Stealth}] (4) -- node[midway, left]{} (6);
    \end{tikzpicture}
\end{equation*}
\subsection{Systems related to the Ablowitz-Ladik lattice}\label{Appb-Q1}
The Ablowitz-Ladik equation \eqref{AL} and the systems \eqref{N1} and \eqref{N2} are related by Miura transformations.
First, system \eqref{N1} (in variables $u$, $v$) is linked to system \eqref{N2} (in variables $w$, $v$) by the transformation:
\begin{equation*}
\eqref{N1} \rightarrow \eqref{N2}: \qquad w = u v, \qquad v \mapsto -v.
\end{equation*}
We introduce a new system:
\begin{equation*}
    \left\{\begin{array}{@{}l@{}}
      r_x= 2 r (\mu-s_1)^{-1} r_1 s_1-2s r_{-1}(\mu_{-1}-s)^{-1}r  \\
      s_x = 2  (s r_{-1}-r s)
  \end{array}\right. \tag{${\rm N_1'}$}\label{N1p}
\end{equation*}
The evolution equation for $s$ shows that \eqref{N1p} is analogous to the relativistic Toda lattice. This system is derived from \eqref{N1} via the substitution:
\begin{equation*}
\eqref{N1} \rightarrow \eqref{N1p}: \qquad r = -v (u+\mu), \qquad s =  - u_{-1}.
\end{equation*}
The Ablowitz-Ladik lattice \eqref{AL} (in variables $u$, $v$) is related to \eqref{N1p} by the Miura transformation:
\begin{equation*}
\eqref{AL} \rightarrow \eqref{N1p}: \qquad r = v_{-1} \mu u, \qquad s_1 = \nu - v u, \qquad \nu \mapsto \mu.
\end{equation*}
The relations among the systems \eqref{AL}, \eqref{N1} and \eqref{N2} are summarised by the following graph:
\begin{equation*}
    \begin{tikzpicture}[main/.style = {draw, rectangle}, none/.style={}, none1/.style={draw, shape = circle, fill = black, minimum size = 0.05cm, inner sep=0pt}] 
    \node[main] (0) {\eqref{N1}};
    \node[none] (01)[right of=0] {};
    \node[main] (1)[right of=01] {\eqref{N1p}};
    \node[none] (10)[right of=1] {};
    \node[main] (11)[right of=10] {\eqref{AL}};
    \node[none] (02)[left of=0] {};
    \node[main] (4)[left of=02] {\eqref{N2}};
    \draw[Stealth-Stealth] (0) -- node[midway, left]{} (1);
    \draw[{Stealth}-] (1) -- node[midway, left]{} (11);
    \draw[Stealth-Stealth] (0) -- node[midway, left]{} (4);
    \end{tikzpicture}
\end{equation*}
\subsection{Systems related to the Merola-Ragnisco-Tu lattice}\label{Appb-Q2}
The Merola-Ragnisco-Tu \eqref{MRT} equation expressed in terms of the variables $r,s$ is related to the $\cI$ involution, defined by \eqref{involp}, of \eqref{N3}, in terms of $u,v$, by the following Miura transformation
\begin{equation*}
\eqref{MRT} \rightarrow \eqref{N3}: \qquad u_1 =   - s_1^{-1}\nu s, \qquad v = r s,\qquad  \mu \mapsto \mu_{-1}.
\end{equation*}
The system \eqref{N3} can be transformed into
\begin{equation*}
    \left\{\begin{array}{@{}l@{}}
     w_x = 2(w v_1 - v w + w v_1^{-1} w_1 \mu_1 - \mu_{-1} w_{-1}v^{-1} w)\\
     v_x = 2(w \mu - \mu_{-1}w_{-1})
    \end{array}\right. \label{N3a} \tag{${\rm N'_3}$}
\end{equation*}
Note that the dynamic of $v$ is linear and depends on $w$ only. This system is obtained from \eqref{N3} via the substitution
\begin{equation*}
 \eqref{N3} \rightarrow \eqref{N3a}: \qquad  w =  v u,  \qquad  v \mapsto v.
\end{equation*}
From the system \eqref{N3}, we derive the $\cI$ involution of system \eqref{K} written in terms of the variables $r,s$:
\begin{equation*}
    \left\{\begin{array}{@{}l@{}}
         r_x = 2(\mu_{-1}r_{-1} - r \mu)(r + s)\\
         s_x = 2(r + s)(\mu s - s_1 \mu_1)
    \end{array}\right.
\end{equation*}
via the following Miura transformation
\begin{equation*}
    \eqref{K} \rightarrow \eqref{N3}: \qquad  u =  -(r+s),  \qquad v = r \mu - \mu_{-1} r_{-1}.
\end{equation*}
The relations among the equations discussed above are summarised by the following graph:
\begin{equation*}
    \begin{tikzpicture}[main/.style = {draw, rectangle}, none/.style={}, none1/.style={draw, shape = circle, fill = black, minimum size = 0.05cm, inner sep=0pt}] 
    \node[main] (0) {\eqref{K}};
    \node[none] (01)[right of=0] {};
    \node[main] (1)[right of=01] {\eqref{N3}};
    \node[none] (10)[right of=1] {};
    \node[main] (2)[right of=10] {\eqref{MRT}};
    \draw[Stealth-] (1) -- node[midway, left]{} (0);
    \draw[Stealth-] (1) -- node[midway, left]{} (2);
    \end{tikzpicture}
\end{equation*}
\bibliographystyle{abbrv}
\bibliography{biblioDNLS}
\end{document}